\documentclass[runningheads]{llncs}[11pt]
\usepackage{amssymb}
\usepackage{graphicx}
\usepackage{color,xcolor}
\usepackage{amsmath,mathrsfs}
\usepackage{hyperref}

\usepackage{url}
\urldef{\mailsa}\path|{alfred.hofmann, ursula.barth, ingrid.haas, frank.holzwarth,|
\urldef{\mailsb}\path|anna.kramer, leonie.kunz, christine.reiss, nicole.sator,|
\urldef{\mailsc}\path|erika.siebert-cole, peter.strasser, lncs}@springer.com|

\makeatletter

\newcommand{\Rmnum}[1]{\expandafter\@slowromancap\romannumeral #1@}
\makeatother

\def\F{\mathbb{F}}

\DeclareMathOperator{\wt}{wt_H}

\newcommand{\supp}{ {\rm Supp}}

\def\F{\mathbb{F}}

\usepackage{makecell}
\usepackage{float}

\begin{document}

\mainmatter  % start of an individual contribution

% first the title is needed
\title{Several new classes of self-orthogonal minimal linear codes violating the Ashikhmin-Barg condition}

% a short form should be given in case it is too long for the running head
\titlerunning{ Self-orthogonal minimal linear codes violating the Ashikhmin-Barg condition}

% the name(s) of the author(s) follow(s) next
%
% NB: Chinese authors should write their first names(s) in front of
% their surnames. This ensures that the names appear correctly in
% the running heads and the author index.
%

\author{Wengang Jin
\and  Kangquan Li
\and Longjiang Qu\thanks{Corresponding author. Email: ljqu$\_$happy@hotmail.com.}
}
\authorrunning{Wengang Jin, Kangquan Li, Longjiang Qu}

\institute{College of Science, National University of Defence Technology \\
Changsha, Hunan 410000, P.R. China. \\
\email{jinwengang0110@126.com,likangquan11@nudt.edu.cn.}
}

%
% NB: a more complex sample for affiliations and the mapping to the
% corresponding authors can be found in the file "llncs.dem"
% (search for the string "\mainmatter" where a contribution starts).
% "llncs.dem" accompanies the document class "llncs.cls".
%

\toctitle{Several new classes of self-orthogonal minimal linear codes violating the Ashikhmin-Barg condition}
\tocauthor{}
\maketitle

\begin{abstract}

Linear codes have attracted considerable attention in coding theory and cryptography due to their significant
applications in secret sharing schemes, secure two-party computation, Galois geometries, among others. As two special subclasses of linear codes, minimal linear codes and self-orthogonal linear codes are of particular interest. Constructing linear codes that possess both minimality and self-orthogonality is very interesting.  The main purpose of this paper is to construct self-orthogonal minimal linear codes that violate the Ashikhmin-Barg (AB for short) condition over the finite field $\mathbb{F}_p$. First, we present several classes of self-orthogonal minimal linear codes violating the AB condition over the finite field $\mathbb{F}_2$ and determine their weight distributions. Next, for any odd prime $p$, we construct two classes of self-orthogonal linear codes from $p$-ary functions, which contain some optimal or almost optimal codes.  Finally, based on plateaued functions, we construct two classes of self-orthogonal linear codes that violate the AB condition. Their weight distributions are also provided. To the best of our knowledge, this paper is the first to investigate the constructions of linear codes that violate the AB condition and satisfy self-orthogonality.

\end{abstract}

%{\bf Mathematics Subject Classification: }
%06E30, 05A05, 05A10, 35F05,  11T06,  11T71,  11T55.

\section{Introduction}\label{introduction}
Bent functions, introduced by Rothaus \cite{Rothaus-1976} in 1976, are maximally nonlinear Boolean functions with an even number of variables. To overcome the restriction of bent functions being not balanced and existing only with an even number of variables, the concept of semi-bent function has been introduced by Chee et al. \cite{Chee-Lee-Kim-1995}. Later, Zheng et al.\cite{Zheng-Zhang-1999(1),Zheng-Zhang-1999(2)} generalized semi-bent functions to the so-called plateaued functions, which are characterized by the property that their Walsh spectra are three-valued (more precisely $0, \pm 2^r$ for a positive integer $r$). Plateaued and bent functions can be used in many domains. Specifically, they have been widely used in designing good linear codes for several applications (such as secret sharing and two-party computation), partial difference sets, association schemes, and strongly regular graphs. Notably, plateaued and bent functions have been widely used to construct linear codes with few weights \cite{Mesnager-2017,Mesnager-2019,Mesnager-2020,MS}.

Two generic constructions of linear codes involving functions over finite fields are called the first and second generic constructions, respectively. The first generic construction (Construction 1 for short, see Eqs. (\ref{code_1}) to (\ref{lcode_1})) is based on the trace function of some functions involving polynomials that vanish at zero \cite{1975-Delsarte}, while the second generic construction is based on the defining set
\cite{Baumert-1972}. As a particular class of linear codes, codes with few weights or, more significantly, those minimal linear codes have many applications. For example, they can be used in communication \cite{Huffman-Pless-2003}, designing secret sharing schemes (see e.g. \cite{CDY,DD,MS,S3,YD} and the references therein), secure two-party computation \cite{CCP}, association schemes \cite{Calderbank-1984}, authentication codes \cite{DW} and strongly regular graphs \cite{Calderbank-1986}. Furthermore, they could be decoded with a minimum distance decoding method \cite{Ashikhmin-Barg-1988}. The construction of minimal linear codes, especially those that violate the AB condition, is quite challenging. Since 2018, many minimal linear codes that violate the AB condition have been constructed in the literature \cite{Heng-Ding-Zhou2018,Sihem-2020,Sinak-2020,2019-Xu-Qu,2022-Xu-Qu,2020-Xu-Qu}.

Self-orthogonal codes are a very significant subclass of linear codes, as they have diverse nice applications in various fields, such as quantum codes, lattices, and linear complementary dual codes (LCD codes for short). Recently, some infinite families of self-orthogonal codes were constructed in the literature \cite{Heng-Li-Liu2023,Heng-Li-Wu2024,Heng-Li2024,Wang-Heng2024,Zhou-Li-Tang-2018}, and their applications in quantum codes and LCD codes were also studied. The construction of these linear codes is a fascinating research topic in error-correcting code theory.

Divisibility is an important property of linear codes. A linear code $\mathcal{C}$ over the finite field $\mathbb{F}_{q}$, where $q$ is a prime power, is said to be divisible if all its codewords have weights divisible by an integer. Divisible codes have diverse applications, such as Galois geometries, subspace codes, partial spreads, vector space partitions, as well as Griesmer codes \cite{Kiermaier-Kurz2020,Kurz2022,Ward1998,Ward2001}. In coding theory, divisibility and self-orthogonality of the liner codes $\mathcal{C}$ over $\mathbb{F}_{q}$ are closely related, and the following results are known.
\begin{itemize}
\item If $\mathcal{C}$ is a binary self-orthogonal code over $\mathbb{F}_{2}$, then the weight of each codeword is divisible by two. If $\mathcal{C}$ is a binary code over $\mathbb{F}_{2}$ and the weight of each codeword in $\mathcal{C}$ is divisible by four, then $\mathcal{C}$ is self-orthogonal \cite[Theorem 1.4.8]{Huffman-Pless-2003}. For convenience, throughout this paper, binary codes whose codewords all have weights divisible by four are called doubly-even. A self-orthogonal code must be even, one which is not doubly-even is called singly-even.
\item $\mathcal{C}$ is a ternary self-orthogonal code over $\mathbb{F}_{3}$ if and only if the weight of each codeword in $\mathcal{C}$ is divisible by three \cite{Huffman-Pless-2003}.
\item Let $\mathcal{C}$ be an $[n, k, d]$ linear code over $\mathbb{F}_q$ with $\mathbf{1}\in\mathcal{C}$, where $\mathbf{1}$ is the all-1 vector of length $n$. If $\mathcal{C}$ is $p$-divisible, then $\mathcal{C}$ is
self-orthogonal,  where $p$ is an odd prime \cite[Theorem 3.2]{Heng-Li2024}.
\end{itemize}

Inspired by the above literature, we have three main motivations for writing this paper.
\begin{itemize}
\item Will it be possible to provide a necessary and sufficient condition for determining whether a class of binary linear codes is self-orthogonal? If $\mathcal{C}$ is a singly-even code, then how to determine whether it is self-orthogonal?
\item For an odd prime $p$, will it be possible to provide a necessary and sufficient condition for determining whether a class of $p$-ary linear codes is self-orthogonal?
\item Since minimal codes and self-orthogonal codes are of such importance, there arises naturally a question: is it possible to design a class of linear codes that are simultaneously self-orthogonal and minimal, yet violate the AB condition?
\end{itemize}

Based on the above motivations, our main contributions in this paper can be summarized as follows. In
Theorem \ref{lfunction1}, we establish the necessary and sufficient condition for the binary linear codes designed by Construction 1 to be self-orthogonal. Two classes of  minimal binary linear codes violating the AB condition are presented in Theorems \ref {lt_1} and \ref {ltt_1}, respectively. Subsequently, we judge that they are self-orthogonal through Theorem \ref{lfunction1}. It should be noted that in some cases, the self-orthogonality of these codes cannot be determined through the conclusions of  \cite{Huffman-Pless-2003}. A detailed comparison and explanation are presented in Remark \ref{remark2}.
In Theorem \ref{bent function1}, we provide the necessary and sufficient condition for the $p$-ary linear codes designed  by Construction 1 to be self-orthogonal, where $p$ is an odd prime. Two classes of $p$-ary linear codes are obtained in Theorem \ref{ccwwl}, which determined by Theorem \ref{bent function1} are self-orthogonal under certain cases. Furthermore, in Remark \ref{ttremark} we clearly specify the cases under which they are self-orthogonal and those under which they are not. It should be noted that the result in \cite[Theorem 3.2]{Heng-Li2024} could not achieve this. In Theorem \ref{ccww}, we construct two classes of $p$-ary minimal linear codes that are self-orthogonal and violate the AB condition. These linear codes are obtained by unbalanced plateaued functions, and the self-orthogonality is derived from Theorem \ref{bent function1}.

The paper is organized as follows. In Section \ref{Preliminaries}, we fix some notation and present necessary background related to the theory of $p$-ary (plateaued) functions and coding theory.
In Section \ref{MM-B-class}, two classes of self-orthogonal minimal binary linear codes violating the AB
condition are provided, their weight distributions are also completely determined. Next, in Section \ref{M-B-class}, we present two classes of self-orthogonal linear codes with two and four weights, as well as two classes of self-orthogonal minimal $p$-ary linear codes violating the AB condition with four to six weights from plateaued functions.
 We shall present several explicit examples of binary and $p$-ary linear codes. Finally, Section \ref {Sec-Conclusion} concludes the paper.

\section{Notation and preliminaries}\label{Preliminaries}
Throughout this paper we adopt the following notation unless otherwise stated.
\begin{itemize}
\item For any finite set $E$, $E^{\star}=E \setminus\{0\}$ and $\# E$
denotes the cardinality of $E$;
\item For any complex number $z$, $|z|$ denotes its modulus;
\item $p$ is a prime and $q=p^n$;
\item $\mathbb {F}_{p^n}$ is the Galois field of order $p^n$;
\item For any $k$ dividing $n$, $\mathrm{Tr}_k^n$ denotes the trace function from $\mathbb{F}_{p^n}$ onto $\mathbb{F}_{p^k}$
$\mathrm{Tr}_k^n(x)=x+x^{p^k}+\cdots+ x^{p^{n-k}}$; for $k=1$, $\mathrm{Tr}_1^n(x)=\sum\limits_{i=0}^{n-1}x^{p^i}$;
\item $\mathbb{Z}$ is the rational integer ring, $\mathbb{Q}$ is the rational field and ${\Bbb C}$ is the complex field;
\item $p^*=(-1)^{\frac{p-1}{2}}p$;
\item $\zeta_p=e^{\frac{2\pi i}{p}}$ is a primitive $p$-th root of unity, where $ i=\sqrt{-1}$ denotes a complex primitive $4$-th root of unity;
\item $\eta$ and $\eta_0$ are the quadratic characters of $\mathbb{F}_q^{\star}$ and  $\mathbb{F}_p^{\star}$, respectively;
\item $\textnormal{SQ}$ and $\textnormal{NSQ}$ denote the sets of all  squares and nonsquares in $\mathbb{F}_p^{\star}$, respectively.
\end{itemize}

\subsection{Cyclotomic field $Q(\zeta_p)$}

In this subsection we recall some basic results on cyclotomic fields (see e.g. \cite{IrelandRosen1990}), which will be used
to calculate the weight distributions of the proposed linear codes in the sequel.

The ring of integers in $\mathbb{Q}\left(\zeta_p\right)$ is $\mathcal {O}_K=\mathbb{Z}[\zeta_p]$. An integral basis  of $\mathcal {O}_{\mathbb{Q}\left(\zeta_p\right)}$ is $\left\{\zeta_p^i \mid 1\leq i\leq p-1\right\}$. The field extension $\mathbb{Q}(\zeta_p)/\mathbb{Q}$ is Galois of degree $p-1$ and the Galois group $Gal \left(\mathbb{Q}\left(\zeta_p\right)/\mathbb{Q}\right)=\left\{\sigma_a \mid a \in \mathbb{F}_p^{\star}\right\}$, where the automorphism $\sigma_a$ of $\mathbb{Q}\left(\zeta_p\right)$ is defined by $\sigma_a(\zeta_p)=\zeta_p^a$. The field $\mathbb{Q}(\zeta_p)$ has a unique quadratic subfield $\mathbb{Q}(\sqrt{p^*})$. For $1\leq a\leq p-1$, $\sigma_a\left(\sqrt {p^*}\right)=\eta_0(a)\sqrt{p^*}$. Hence, the Galois group $Gal(\mathbb{Q}(\sqrt{p^*})/\mathbb{Q})$ is $\left\{1, \sigma_{\gamma}\right\}$, where $\gamma$ is any quadratic nonresidue in $\mathbb{F}^\star_p$. Finally,  for any $a\in \mathbb{F}_p^{\star}$ and $b\in \mathbb{F}_p$, one immediately has $\sigma_a\left(\zeta_p^b\right)=\zeta_p^{ab}$ and $\sigma_a\left(\sqrt{p^*}^m\right) =\eta^m_0(a)\sqrt{p^*}^m$.

\subsection{Gauss sums and character sums}
A character $\phi$ of a finite abelian group $(\mathcal G, +)$ with order $v$ is a homomorphism from $\mathcal G$ to the
multiplicative group of complex numbers of absolute value $1$. In particular, let $\mathcal G=\mathbb{F}_q$ be the finite field of $q$ elements, where $q=p^n$, $p$ is a prime, and $n$ is a positive integer. For any $a \in \mathbb{F}_q$, we can define an additive character $\phi_a : \mathbb{F}_q \rightarrow {\Bbb C}^\star$ of the finite field $\mathbb{F}_q$ as
$\phi_a(x)=\zeta_p^{\mathrm{Tr}_1^n(ax)}$.

Let $\chi : \mathbb{F}_q^\star \rightarrow {\Bbb C}^\star$ be a multiplicative
character of $\mathbb{F}_q^\star$. The Gauss sum is defined by
$G_q(\chi,\phi_a)=\sum_{x \in \mathbb{F}_q^\star } \chi(x) \phi_a(x).$ In the absence of ambiguity, $G_q(\chi,\phi_a)$ is abbreviated as $G(\chi,\phi_a)$.
Then from \cite{Lidl-Niederreiter-1983} we have that $G_q(\chi,\phi_a)=\overline{\chi(a)}G_q(\chi,\phi_1)$ for $ a \in \mathbb{F}_q^{\star}$, where the ``overline" denotes the complex conjugation.
Generally, the explicit determination of Gauss sums is a difficult problem. However, they
can be explicitly evaluated in a few cases.

We now recall some useful properties of quadratic Gauss sum over $\mathbb{F}_{p}$, which will be used to derive some results in the core of the article.
\begin{lemma}\label{lequations}\textnormal{\cite{Lidl-Niederreiter-1983}} Let $\phi_1$ be a canonical additive character of $\mathbb{F}_{p^n}$ and $p$ be an odd prime. Then the following statement holds.
\begin{eqnarray*}
G_{p^n}(\eta,\phi_r)=\sum\limits_{x\in\mathbb{F}_{p^n}^{\star}}\eta(x)\zeta_{p}^{\mathrm{Tr}_{1}^{n}(rx)}&=&
(-1)^{n-1}\eta(r)\sqrt{p^{*}}^{n}.
\end{eqnarray*}
\end{lemma}
\begin{lemma}\label{equations}\textnormal{\cite{Lidl-Niederreiter-1983}} Let $\phi_1$ be a canonical additive character of $\mathbb{F}_{p}$ and $p$ be an odd prime. Then the following statements hold.
\begin{enumerate}
\item $G_{p}(\eta_0,\phi_r)=\sum\limits_{x\in\mathbb{F}_{p}^{\star}}\eta_0(x)\zeta_{p}^{rx}=\eta_0(r)\sqrt{p^{*}}$, $r\in \mathbb{F}_{p}^{\star}.$
\item $\sum\limits_{x \in \mathbb{F}_p}\zeta_{p}^{f(x)} = \zeta_{p}^{a_0 - a_1^2(4a_2)^{ - 1}}G_{p}(\eta_0,\phi_{a_2})$,  where   $f(x) = a_2x^2 + a_1x + a_0 \in \mathbb{F}_p[x]$ and $a_2 \not= 0$.
\end{enumerate}
\end{lemma}
\begin{lemma}\textnormal{\cite{Lidl-Niederreiter-1983}}\label{squar-Fact} Let $g(x)=a_2x^2+a_1x+a_0\in \mathbb{F}_{p}[x]$ and $a_2\ne0$. Put $d=a^2_1-4a_0a_2$. Then
\begin{eqnarray*}
\sum\limits_{x\in \mathbb{F}_{p}}\eta_0(g(x))=
\left\{ \begin{array}{ll}
-\eta_0(a_2),& \mathrm{if}~d\ne0,  \\
(p-1)\eta_0(a_2),& \mathrm{if}~d=0.
\end{array} \right.
\end{eqnarray*}
\end{lemma}
\subsection{On $p$-ary functions}
A $p$-ary function is a function from the finite field $\F_{p^n}$ to $\F_p$. Recall that all functions from $\F_{p^n}$ to $\F_{p^n}$
(or with values in a subfield) have a unique representation as a univariate polynomial of degree at most
$p^n-1$. A $p$-ary function $f$ can be uniquely represented
as a polynomial $f(x) = \sum_{j=0}^{p^n-1}a_jx^j$ of polynomial degree at
most $p^n-1$. Every exponent $j$ can be written in base $p$ representation (or the $p$-adic representation) $j = \sum_{i=0}^{n-1}t_ip^i$ with $0\leq t_i \leq p-1$,
the {\it weight} of $j$ is then $\sum_{i=0}^{n-1}t_i$. The {\it algebraic degree} of $f$ equals the largest weight
of an exponent $j$ in the polynomial representation of $f$, for which $a_j\not=0$.
\subsection{Walsh transform  and $p$-ary (weakly regular) plateaued functions}
Let $\zeta_p$ be the $p$-th  complex root of unity. The following basic fact related $\zeta_p$
\begin{eqnarray}\label{othor-prop-wp}
\sum\limits_{x\in \mathbb{F}_{p}^{\star}}\zeta_{p}^{ax}&=&
\left\{ \begin{array}{ll}
p-1, &\mathrm{if~} a=0,  \\
-1, &\mathrm{if~} a\in \mathbb{F}_{p}^{\star},
\end{array} \right.
\end{eqnarray}
will be frequently used in the sequel.

Let $f$ be a function from $\mathbb{F}_{q}$ to $\mathbb{F}_p$, where $q=p^n$ is a prime power. The  Walsh transform of $f$ at $\beta\in\mathbb{F}_{q}$ is defined by
$$
\mathcal{W}_{f}(\beta)=
\sum_{x\in\mathbb{F}_{q}}\zeta_p^{f(x)
-\mathrm{Tr}_{1}^{n}(\beta x)}.
$$
 %The inverse Walsh transform of such $f$
%gives
%\begin{equation}\label{equ1}
%\zeta_p^{f(x)}=\frac{1}{p^n}
%%\sum_{\beta\in \mathbb{F}_q}
%\mathcal{W}_f(\beta)\zeta_p^{\mathrm{Tr}_1^n(\beta x)}.
%\end{equation}

The multiset $\{\mathcal{W}_f(\beta)\,|\,\beta\in\mathbb{F}_{p^n}\}$ is called the {\it Walsh spectrum} of the function
$f$. The function $f$ is said to be $p$-ary bent, if $|\mathcal{W}_f(\beta)|=p^{\frac{n}{2}}$ for any $\beta\in \mathbb{F}_{p^n}$.
Bent functions from $\mathbb{F}_{p^n}$ to $\F_p$, $p$ odd, exist
for all integers $n\ge 1$.
As first shown by Kumar et. al \cite{Kumar-Scholtz-Welch-1985} (see also \cite{Helleseth-Kholosha-2006}), the values of the Walsh transform of a $p$-ary bent function are restricted. More precisely,
for a $p$-ary bent function $f:\mathbb{F}_{p^n}\rightarrow\F_p$, the \textit{Walsh coefficient} $\mathcal{W}_f(\beta)$
at $\beta\in\mathbb{F}_{p^n}$ always satisfies
\begin{equation}\label{Eq1}
\mathcal{W}_f(\beta) =
\left\{\begin{array}{r@{\quad}l}
\pm \zeta_p^{f^*(\beta)}p^{n/2}, &\mathrm{if~} p\equiv 1 \pmod 4, \\
\pm i\zeta_p^{f^*(\beta)}p^{n/2}, &\mathrm{if~} p\equiv 3 \pmod 4,
\end{array}\right.
\end{equation}
where $f^*$ is a function from $\mathbb{F}_{p^n}$ to $\F_p$, which  is called the \emph{dual }of $f$ $\left(i = \sqrt{-1}\right)$.

A  bent function $f$ is said to be regular if there exists some $p$-ary function $f^*$ satisfying $\mathcal{W}_f(\beta)=p^{\frac{n}{2}}\zeta_p^{f^*(\beta)}$
for any $\beta \in \mathbb{F}_{q}$.
A  bent function $f$ is weakly regular if
there exists a complex $u$ with unit magnitude
satisfying
$\mathcal{W}_f(\beta)=up^{\frac{n}{2}}\zeta_p^{f^*(\beta)}$
for some function $f^*(x)$, where $u\in\{\pm1, \pm i\}$ is independent from $\beta$, otherwise it is called non-weakly regular.

As an extension of bent functions, the notion of plateaued functions was introduced in characteristic
2 by Zheng and Zhang in \cite{Zheng-Zhang-1999(1)}. A function $g: \mathbb{F}_{p^n}\to \mathbb{F}_p$ is said to be {\it $s$-plateaued} if $\mid\mathcal{W}_g(\beta)\mid^{2}\in\left\{0,p^{n+s}\right\}$ for every $\beta \in\mathbb{F}_{p^n},$ where $s$ is an integer with $0\leq s \leq n$. Furthermore, the {\it Walsh support} of an $s$-plateaued function $g$ is defined by
$$\supp(\mathcal{W}_g)=\left\{\beta\in\mathbb{F}_{p^n}~|~\mid\mathcal{W}_g(\beta)\mid^{2}=p^{n+s}\right\}.$$
A plateaued function $f$ is said to be {\it balanced} if $\mathcal{W}_f(0)=0$; otherwise, $f$ is called {\it unbalanced}.
\begin{lemma}\label{tffunction g}\textnormal{\cite{Mesnager-2019}}
Let $g: \mathbb{F}_{p^n}\to \mathbb{F}_p$ be an $s$-plateaued function and $s$ be an integer with $0\leq s \leq n$. Then for $\beta\in\mathbb{F}_{p^n}$, $\mid\mathcal{W}_g(\beta)\mid^{2}$ takes $p^{n-s}$ times the value $p^{n+s}$ and $p^{n}-p^{n-s}$ times the value $0$.
\end{lemma}

The non-trivial subclass of plateaued functions was introduced by Mesnager
et al. \cite{Mesnager-2019}.
\begin{definition}
Let $g: \mathbb{F}_{p^n}\to \mathbb{F}_p$ be an $s$-plateaued function and $s$ be an integer with $0\leq s \leq n$. Then $f$ is called weakly regular s-plateaued if there exists a complex number $u$ having unit magnitude
such that
$$\mathcal{W}_g(\beta)\in\left\{0, up^{\frac{n+s}{2}}\zeta^{g^{*}(\beta)}_p\right\} $$
for all $\beta\in\mathbb{F}_{p^n}$, with $g^*$ being a $p$-ary function over $\mathbb{F}_{p^n}$ and $g^*(\beta)=0$ for all $\beta\in \mathbb{F}_{p^n}\setminus \supp(\mathcal{W}_g)$. Otherwise $g$ is called non-weakly regular $p$-ary $s$-plateaued. Note that a weakly regular function $g$ is said to be regular if $u = 1$. It is obvious that a weakly regular plateaude function
$f$ satisfies
\begin{equation*}
\mathcal{W}_g(\beta)=\epsilon\sqrt{p^*}^{n+s} \zeta_p^{g^*(\beta)},
\end{equation*}
where $\epsilon =\pm 1$ is called the sign of the Walsh transform of $g$.
\end{definition}
%%%%%%%%%

In \cite{Tang-Li-Qi-Zhou-Helleseth-2016}, Tang, Li, Qi, Zhou, and Helleseth introduced a very interesting  special class $\mathcal{RF}$ of $p$-ary weakly regular bent functions.
Precisely, functions in the class $\mathcal{RF}$ are $p$-ary weakly regular bent and satisfy the following conditions:

${\rm (P1)}$ $f(0)=0;$

${\rm (P2)}$ There exists an integer $l$ such that $(l-1,p-1)=1$ and $f(cx)=c^{l}f(x)$ for any $c\in \mathbb{F}_p^{\star}$ and $x\in \mathbb{F}_{p^n}^{\star}.$

From  \cite{Tang-Li-Qi-Zhou-Helleseth-2016},
it has been proved that the dual $f^*\in \mathcal{RF}$ for any $f\in \mathcal{RF}$. In Subsection \ref{p-p-class}, we always assume that the dual $f^*$ of a weakly regular plateaued function $f$ belongs to $\mathcal{RF}$.

The next lemma shall be used frequently in Section \ref{MM-B-class}.
\begin{lemma}\label{function g}\textnormal{\cite{Helleseth-Kholosha-2006}}
Let $p$ be a prime. Let $n=2m$ and $\lambda \in \mathbb{F}_{p^{m}}^{\star}$. Then the $p$-ary monomial $g_{\lambda}(x)=\mathrm{Tr}^{m}_{1}\left(\lambda x^{p^{m}+1}\right)$ is a weakly regular bent function. Moreover, for all $a\in\mathbb{F}_{p^n},$ the corresponding Walsh transform coefficient of $g_{\lambda}(x)$ is %equal to \mathcal{W}_{g_\lambda}(a)
$$ \mathcal{W}_{g_\lambda}(a)=-p^m\zeta_{p}^{-\mathrm{Tr}^{m}_{1}\left(\lambda^{-1}a^{p^m+1}\right)}.$$
\end{lemma}
\subsection{Some background related to coding theory}

The \emph{Hamming weight} of a codeword ${\mathbf{a}}=(a_0,a_1,\cdots,a_{n-1})$ of $\mathbb{F}_q$, denoted by $\wt(\mathbf{a})$, is the cardinality of its \emph{support} defined as Supp$(\mathbf{a})=\{0\leq i\leq n-1~|~ a_i\neq 0\}$. An $[n, k, d]$ linear code $\mathcal C$ over $\mathbb{F}_q$ is a $k$-dimensional vector subspace of the $n$-dimensional vector space $\mathbb{F}_{p^n}$ with minimum nonzero Hamming weight $d$, where $p$ is a prime and $n$ is a positive integer. An element of $\mathcal C$ is said to be a \emph{codeword}.
The \emph{weight enumerator} of $\mathcal C$ with length $n$ is the polynomial $1+A_1z+A_2z^2+\cdots+A_{n}z^n$, where $A_i$ denotes the number of codewords with Hamming weight $i$ in $\mathcal C$. The sequence $(1, A_1, A_2,\cdots, A_n)$ is called the \emph{weight distribution} of $\mathcal C$. A code $\mathcal C$ is said to be \emph{$t$-weight} if the number of nonzero $A_i$ in the
sequence $(A_1, A_2,\cdots, A_n)$ is equal to $t$.

Given a linear code $\mathcal C$ of length $n$ over $\mathbb F_{q}$, its (Euclidean) dual code, denoted by $\mathcal C^ {\perp}$, is defined as
\begin{align*}
\mathcal C^{\perp}=\Biggl\{ \mathbf{b}=(b_0, b_1,\cdots, b_{n-1})\in \mathbb F_{q}^n~|~ & \mathbf{b}\cdot\mathbf{c}=\sum _{i=0}^{n-1} b_i  c_i=0,\\ \forall~ \mathbf{c}=(c_0, c_1, \cdots, c_{n-1}) \in \mathcal C \Biggl\}.
\end{align*}
If $\mathcal{C} \subseteq\mathcal{C}^{\perp}$, then $\mathcal{C}$ is called a \emph{self-orthogonal code}, that is, $\mathbf{a}\cdot\mathbf{b}=0$ for any codewords $\mathbf{a},\mathbf{b}\in\mathcal{C}$.  In particular, if $\mathcal{C} = \mathcal{C}^{\perp}$, then $\mathcal{C}$ is said to be \emph{self-dual}. Self-orthogonal codes are a very important subclass of linear codes as they have many nice applications.

A codeword $\mathbf{b}$ of $\mathcal C$ is covered by another codeword $\mathbf{a}$ of $\mathcal C$ if Supp($\mathbf{a}$) contains Supp($\mathbf{b}$). A nonzero codeword $\mathbf{a}\in \mathcal{C}$ is called \emph{minimal} if it only covers the codewords $c\mathbf{a}$ for all $c\in \mathbb{F}_p$, but no other nonzero codewords of $\mathcal C$. A linear code $\mathcal C$ is called \emph{minimal} if every nonzero codeword of $\mathcal C$ is minimal. The following lemma shows that a linear code $\mathcal C$ is minimal if its weights are close enough to each other, which was proved by Ashikhmin and Barg  \cite{Ashikhmin-Barg-1988}.
\begin{lemma}\label{LEM1} \textnormal{\cite{Ashikhmin-Barg-1988}}
Let $\mathcal C$ be a linear code over $\mathbb{F}_p$, $w_{\min}$ and $w_{\max}$ denote the minimum and maximum
Hamming weights of nonzero codewords in $\mathcal C$, respectively. If
$$\frac{p-1}{p}<\frac{w_{\min}}{w_{\max}},$$ then  $\mathcal C$ is minimal.
\end{lemma}

The above minimality condition is sufficient but not necessary, Ding,
Heng and Zhou in \cite{DHZ18,Heng-Ding-Zhou2018} derived a necessary and sufficient condition on weights for a given linear code to be minimal as follows.
\begin{lemma}\label{nsc}\textnormal{\cite{DHZ18,Heng-Ding-Zhou2018}}
Let $\mathcal C$ be a linear code over $\mathbb{F}_p$. Then $\mathcal C$ is minimal if and only if
$$\sum\limits_{z \in \mathbb{F}^{\star}_p}\wt(\mathbf{a}+z\mathbf{b})\ne(p-1)\wt(\mathbf{a})-\wt(\mathbf{b})$$
for any linearly independent codewords $\mathbf{a},\mathbf{b}\in\mathcal C$.
\end{lemma}
\begin{lemma}\label{2b}\textnormal{\cite{Huffman-Pless-2003}} $(Singleton~bound)$  If a $q$-ary $[n,k,d]$ code exists, then
$$k\leq n-d+1.$$
\end{lemma}
\begin{lemma}\label{2bt}\textnormal{\cite{Huffman-Pless-2003}} $(Griesmer~bound)$ If a $q$-ary $[n,k,d]$ code exists, then
$$n\geqslant\sum_{i=0}^{k-1}\lceil\frac{d}{q^{i}}\rceil.$$
Here, $\lceil\cdot\rceil$ denotes the ceiling function.
\end{lemma}

An $[n, k, d]$ code is called optimal if there is no $[n, k, d+1]$ code. An $[n, k, d]$ code is said to be almost optimal if an $[n, k, d+1]$ code is optimal.

Let $p$ be a prime and $f(x)$ be a function from $\mathbb{F}_{p^n}$ to $\mathbb{F}_{p}$  and $f(x)\ne\mathrm{Tr}^{n}_{1}(\omega x)$ for any $\omega\in\mathbb{F}_{p^n}.$ Based on the $p$-ary function $f(x)$, there exist several classes of classical constructions of linear codes below:
\begin{equation}\label{code_1}
\mathcal{C}^{\star}_{f} =\left\{\mathbf{c}_{\alpha,\beta}=(\alpha f(x)-\mathrm{Tr}_1^n(\beta x))_{x\in\mathbb{F}^{\star}_{ p^{n}}}~|~\alpha \in \mathbb{F}_{p},\beta\in \mathbb{F}_{p^n}\right\}.
\end{equation}
\begin{equation}\label{vzcode_1}
\mathcal{C}_{f} =\left\{\mathbf{c}_{\alpha,\beta}=(\alpha f(x)-\mathrm{Tr}_1^n(\beta x))_{x\in\mathbb{F}_{ p^{n}}}~|~\alpha \in \mathbb{F}_{p},\beta\in \mathbb{F}_{p^n}\right\}.
\end{equation}
Moreover,the augmented code of extended code of $\mathcal{C}^{\star}_{f}$ is given by
\begin{equation}\label{lcode_1}
\overline{\mathcal{C}}_{f} =\left\{\mathbf{c}_{\alpha,\beta,c}=(\alpha f(x)-\mathrm{Tr}_1^n(\beta x)-c)_{x\in\mathbb{F}_{ p^{n}}}~|~\alpha \in \mathbb{F}_{p},\beta\in \mathbb{F}_{p^n}, c \in \mathbb{F}_{p}\right\}.
\end{equation}

In order to simplify the computation of the weights for the linear codes in Eqs. (\ref{code_1}) to (\ref{lcode_1}), we present the following results.
\begin{lemma}\label{code-lemma}\textnormal{\cite{2019-Xu-Qu}}
Let the codes $\mathcal{C}^{\star}_{f}$ and $\mathcal{C}_{f}$ be defined in \textnormal{Eqs. (\ref{code_1})} and \textnormal{(\ref{vzcode_1})}, respectively. Then $\mathcal{C}^{\star}_{f}$ is a $\left[p^n-1, n+1\right]$ code and $\mathcal{C}_{f}$ is a $\left[p^n, n+1\right]$ code. Furthermore, the Hamming weight of $\mathbf{c}_{\alpha,\beta}$ is
\begin{eqnarray}\label{codea3}
\wt(\mathbf{c}_{\alpha,\beta})
&=&\left\{ \begin{array}{ll}
0, &\mathrm{if~} \alpha=\beta=0,\\
p^{n}-p^{n-1}, &\mathrm{if~}  \alpha=0, \beta\ne0,\\
p^{n}-p^{n-1}-\frac{1}{p}\sum\limits_{z\in\mathbb{F}^{\star}_{p}}\sigma_{z\alpha}\left(\mathcal{W}_{f}\left(\alpha^{-1}\beta\right)\right), & \mathrm{ otherwise}.
\end{array} \right.
\end{eqnarray}
\end{lemma}
\begin{lemma}\label{1code-lemma}
Let $\overline{\mathcal{C}}_{f}$ be the linear code defined by \textnormal{Eq. (\ref{lcode_1})}, where $f(x)$ is a function from $\mathbb{F}_{p^n}$ to $\mathbb{F}_{p}$. Then $\overline{\mathcal{C}}_{f}$ is a $[p^n, n+1]$ code and the Hamming weight of $\mathbf{c}_{\alpha,\beta,c}$ is
\begin{eqnarray}\label{1codea3}
\wt(\mathbf{c}_{\alpha,\beta,c})
&=&\left\{ \begin{array}{ll}
0, &\mathrm{if~} \alpha=\beta=c=0,\\
p^{n}, &\mathrm{if~}  \alpha=\beta=0, c\ne0,\\
p^{n}-p^{n-1}, &\mathrm{if~}  \alpha=0,\beta\ne0,\\
p^{n}-p^{n-1}-\sum\limits_{z\in\mathbb{F}^{\star}_{p}}\sigma_{z\alpha}\left(\mathcal{W}_{f}\left(\alpha^{-1}\beta\right)\zeta_p^{\alpha^{-1}c}\right), & \mathrm{if~} \alpha\ne0.
\end{array} \right.
\end{eqnarray}
\end{lemma}
\begin{proof} For any $(\alpha,\beta, c)\in \mathbb{F}_{p}\times\mathbb{F}_{p^{n}}\times\mathbb{F}_{p}$, by Eq. (\ref{othor-prop-wp}) we have
\begin{eqnarray*}
\wt(\mathbf{c}_{\alpha,\beta,c})&=&p^n-\frac{1}{p}\sum\limits_{x\in \mathbb{F}_{p^{n}}}\sum\limits_{z\in \mathbb{F}_{p}}\zeta_p^{\left(\alpha f(x)-\mathrm{Tr}_1^n(\beta x)-c\right)z}\nonumber \\
&=&p^n-p^{n-1}-\frac{1}{p}\sum\limits_{x\in \mathbb{F}_{p^{n}}}\sum\limits_{z\in \mathbb{F}^{\star}_{p}}\zeta_p^{(\alpha f(x)-\mathrm{Tr}_1^n(\beta x)-c)z}\nonumber\\
&=&\left\{ \begin{array}{ll}
0, &\mathrm{if~} \alpha=\beta=c=0,\\
p^{n}, &\mathrm{if~}  \alpha=\beta=0, c\ne0,\\
p^{n}-p^{n-1}, &\mathrm{if~}  \alpha=0,\beta\ne0,\\
p^{n}-p^{n-1}-\frac{1}{p}\sum\limits_{z\in\mathbb{F}^{\star}_{p}}\sigma_{z\alpha}\left(\mathcal{W}_{f}\left(\alpha^{-1}\beta\right)\zeta_p^{\alpha^{-1}c}\right), & \mathrm{if~} \alpha\ne0.
\end{array} \right.
\end{eqnarray*}
\end{proof}
\section{Self-orthogonal minimal binary linear codes violating the AB condition }\label{MM-B-class}
In this section, we shall first obtain a necessary and sufficient condition for a binary linear code
to be self-orthogonal. Then, based on specific Boolean functions, we will construct two new classes of minimal linear codes violating the AB condition: one class being doubly-even code and the other class being singly-even self-orthogonal code, where the self-orthogonality can be derived by the necessary and sufficient condition mentioned above.
\begin{theorem}\label{lfunction1}
Let $n\geq 3$ be a positive integer and $p=2$. Then $\mathcal{C}_{f}$ is a self-orthogonal linear code if and only if $\mathcal{W}_{f}(\beta)\pm\mathcal{W}_{f}(0)\equiv 0 \pmod 8$ for all $\beta\in\mathbb{F}_{2^n}$.
\end{theorem}
\begin{proof} Let $\mathbf{a}=\mathbf{c}_{\alpha_1,\beta_1}$ and $\mathbf{b}=\mathbf{c}_{\alpha_2,\beta_2}$ be any two codewords in $\mathcal{C}_{f}$, where $(\alpha_i, \beta_i)\in\mathbb{F}_{2}\times\mathbb{F}_{2^n}$ with $i\in\{1,2\}$. Then by Eq. (\ref{code_1}), we get
\begin{eqnarray}\label{azero1}
\mathbf{a}\cdot\mathbf{b}&=&\mathbf{c}_{\alpha_1,\beta_1}\cdot\mathbf{c}_{\alpha_2,\beta_2}=\sum\limits_{x\in \mathbb{F}_{2^n}}\left(\alpha_1f(x)+\mathrm{Tr}^n_1(\beta_1x)\right)\left(\alpha_2f(x)+\mathrm{Tr}^n_1(\beta_2x)\right)\nonumber\\
&=&\sum\limits_{x\in \mathbb{F}_{2^n}}\left(f(x)\left(\alpha_1\alpha_2+\mathrm{Tr}^n_1((\alpha_1\beta_2+\alpha_2\beta_1)x)\right)+\mathrm{Tr}^n_1(\beta_1x)\mathrm{Tr}^n_1(\beta_2x)\right).
\end{eqnarray}
To compute $\mathbf{a}\cdot\mathbf{b}$, we shall divide the discussion into the following two cases according to the values of $\alpha_1\alpha_2$.

$(1)~\alpha_1\alpha_2=0$. We note that the condition $\alpha_1\alpha_2= 0$ contains three cases, that is $\alpha_1=\alpha_2=0$, $\alpha_1=1, \alpha_2=0$ and $\alpha_1=0, \alpha_2=1$, here we only prove the case of $\alpha_1=0, \alpha_2=1$. From Eq. (\ref{azero1}) we get that
\begin{eqnarray}\label{azero2}
\mathbf{a}\cdot\mathbf{b}&=&\sum\limits_{x\in \mathbb{F}_{2^n}}\left(f(x)\mathrm{Tr}^n_1(\beta_1x)+\mathrm{Tr}^n_1(\beta_1x)\mathrm{Tr}^n_1(\beta_2x)\right)\nonumber\\
&=&\sum\limits_{x\in \mathbb{F}_{2^n}}\frac{1-(-1)^{f(x)\mathrm{Tr}^n_1(\beta_1x)}}{2}+\sum\limits_{x\in \mathbb{F}_{2^n}}\prod\limits^2_{i=1}\frac{1-(-1)^{\mathrm{Tr}^n_1(\beta_ix)}}{2}\nonumber\\
&=&2^{n-1}-\frac{1}{2}\sum\limits_{x\in \mathbb{F}_{2^n}}\sum\limits_{j\in \mathbb{F}_{2}}(-1)^{jf(x)}\frac{1+(-1)^{\mathrm{Tr}^n_1(\beta_1x)+j}}{2}+\sum\limits_{x\in \mathbb{F}_{2^n}}\prod\limits^2_{i=1}\frac{1-(-1)^{\mathrm{Tr}^n_1(\beta_ix)}}{2}\nonumber\\
&=&2^{n-1}-\frac{1}{4}\sum\limits_{x\in \mathbb{F}_{2^n}}\left((-1)^{f(x)+\mathrm{Tr}^n_1(\beta_1x)+1}+(-1)^{f(x)}\right)\nonumber\\
&=& \frac{\mathcal{W}_{f}(\beta_1)-\mathcal{W}_{f}(0)}{4}.
\end{eqnarray}
For the other cases, by adopting a similar method, we can obtain
\begin{eqnarray}
\mathbf{a}\cdot\mathbf{b}
&=& \left\{ \begin{array}{ll}
\frac{\mathcal{W}_{f}(\beta_2)-\mathcal{W}_{f}(0)}{4}, & \mathrm{if}~\alpha_1=1, \alpha_2=0,\\
0, & \mathrm{if}~\alpha_1=\alpha_2=0.
\end{array} \right.
\end{eqnarray}

$(2)~\alpha_1=\alpha_2=1$. Again by Eq. (\ref{azero1}) we get that
\begin{eqnarray}\label{azero3}
\mathbf{a}\cdot\mathbf{b}&=&\sum\limits_{x\in \mathbb{F}_{2^n}}\left(f(x)\mathrm{Tr}^n_1((\beta_1+\beta_2)x)+\mathrm{Tr}^n_1(\beta_1x)\mathrm{Tr}^n_1(\beta_2x)+f(x)\right)\nonumber\\
&=&\sum\limits_{x\in \mathbb{F}_{2^n}}\left(\frac{1-(-1)^{f(x)}}{2}+\frac{1-(-1)^{f(x)\mathrm{Tr}^n_1((\beta_1+\beta_2)x)}}{2}+\prod\limits^2_{i=1}\frac{1-(-1)^{\mathrm{Tr}^n_1(\beta_ix)}}{2}\right)\nonumber\\
&=& \frac{\mathcal{W}_{f}(\beta_1+\beta_2)-3\mathcal{W}_{f}(0)}{4}.
\end{eqnarray}

Finally, combining Eqs. (\ref{azero2}) to (\ref{azero3}) and the definition of self-orthogonal linear codes, we can easily obtain that $\mathcal{C}_{f}$ is self-orthogonal if and only if $\mathcal{W}_{f}(\beta)\pm\mathcal{W}_{f}(0)\equiv 0 \pmod 8$ for all $\beta\in\mathbb{F}_{2^n}$.\hfill $\square$
\end{proof}
\subsection{ Doubly-even self-orthogonal minimal codes violating the AB condition}\label{MM-B1-class}
Let $\lambda_1, \lambda_2, \lambda_3$ be any three linearly independent elements in $\mathbb{F}_{2^m}^{\star}$. For an even positive integer $n=2m\geq6,$  we define a new Boolean function as follows.
\begin{equation}\label{f(x)}
f(x)=f_{\lambda_1}(x)f_{\lambda_2}(x)f_{\lambda_3}(x),~~~~x\in\mathbb{F}_{2^{n}},
\end{equation}
where $f_{\lambda_i}(x)=\mathrm{Tr}_1^m\left(\lambda_i x^{2^m+1}\right)$ with $i\in\{1,2,3\}$. Let $f_0(x)=0$. Then by the Walsh transform of $f(x)$ at $\beta\in\mathbb{F}_{2^n}$ and Lemma \ref{function g}, we have
\begin{eqnarray}\label{ht}
\mathcal{W}_{f}(\beta)&=&\sum\limits_{x\in \mathbb{F}_{2^{n}}}(-1)^{\mathrm{Tr}^m_1\left(\lambda_1 x^{2^m+1}\right)
\mathrm{Tr}^m_1\left(\lambda_2x^{2^m+1}\right)\mathrm{Tr}^m_1\left(\lambda_3 x^{2^m+1}\right)+\mathrm{Tr}_1^n\left(\beta x\right)} \nonumber \\
&=&\sum\limits_{x \in \mathbb{F}_{2^n}}\sum\limits_{c_{1},c_{2}\in \mathbb{F}_{2}}\prod\limits^{2}_{j=1}\frac{1+\left(-1\right)^{c_j+\mathrm{Tr}_1^m\left(\lambda_jx^{2^m+1}\right)}}{2}(-1)^{\mathrm{Tr}_1^n(\beta x)+c_1c_2\mathrm{Tr}_1^m\left(\lambda_3x^{2^m+1}\right)}\nonumber \\
&=&\frac{1}{4}\sum\limits_{c_{1},c_{2}\in \mathbb{F}_{2}}\big(\mathcal{W}_{f_{c_1c_2\lambda_3}}(\beta)+\mathcal{W}_{f_{\lambda_1+c_1c_2\lambda_3}}(\beta)
(-1)^{c_1}+\mathcal{W}_{f_{\lambda_2+c_1c_2\lambda_3}}(\beta)(-1)^{c_2}\nonumber \\&&+\mathcal{W}_{f_{\lambda_1+\lambda_2+c_1c_2\lambda_3}}(\beta)(-1)^{c_1+c_2}\big) \nonumber \\
&=&\frac{3}{4}\mathcal{W}_{f_{0}}(\beta)+\frac{1}{4}\mathcal{W}_{f_{\lambda_1+\lambda_2+\lambda_3}}(\beta)+\frac{1}{4}\sum\limits^{3}_{j=1}\mathcal{W}_{f_{\lambda_j}}(\beta)-\frac{1}{4}
\sum\limits_{1\leq j_1<j_2\leq 3}\mathcal{W}_{f_{\lambda_{j_1}+\lambda_{j_2}}}(\beta)\nonumber \\
&=&\left\{ \begin{array}{ll}
3\cdot2^{n-2}-2^{m-2}, & \mathrm{if}~\beta=0, \\
2^{m-2}\big((-1)^{A_{7}+1}+(-1)^{A_{1}+1}+(-1)^{A_{2}+1}\\+(-1)^{A_{4}+1}+(-1)^{A_{3}}+(-1)^{A_{5}}+(-1)^{A_{6}}\big), & \mathrm{if}~\beta\in\mathbb{F}^{\star}_{2^n},
\end{array} \right.
\end{eqnarray}
where $A_{c_1+c_{2}2+c_{3}2^2}=\mathrm{Tr}_1^m\left(\left(\sum\limits^{3}_{j=1}c_j\lambda_{j}\right)^{-1}\beta^{2^m+1}\right)$ with $(c_1,c_2,c_3)\in\mathbb{F}^3_{2}\backslash\{0,0,0\}$.

For the convenience of calculation, we give the following definition.
\begin{definition}\label{zht}
Let $t_1$ represent the relationship between $\left(\sum^{3}_{j=1}\lambda_j\right)^{-1}$ and $\sum^{3}_{j=1}\lambda^{-1}_j$, that is, $\left(\sum^{3}_{j=1}\lambda_j\right)^{-1}=\sum^{3}_{j=1}\lambda^{-1}_j$ if $t_1=0$, otherwise $t_1=1$. Similarly, $(\lambda_{j_1}+\lambda_{j_2})^{-1}=\lambda^{-1}_{j_1}+\lambda^{-1}_{j_2}$ with $1\leq j_1<j_2\leq 3$, if $t_{j_1+j_2-1}=0$, otherwise $t_{j_1+j_2-1}=1$.
\end{definition}

It is obvious that $(t_1,t_2,t_3,t_4)\in\mathbb{F}^{4}_{2}$. Clearly, each $A_i$, $1\leq i\leq 7$, can be represented by $\wt(t_1,t_2,t_3,t_4)+3$ elements from the set $\{A_1, A_2,\cdots, A_7\}$. For example, if  $(t_1,t_2,t_3,t_4)=(1,0,0,0)$, then we can deduce that $A_{3}=A_{1}+A_{2}$, $A_{5}=A_{1}+A_{4}$, $A_{6}=A_{2}+A_{4}$ and $A_{7}=A_{7}$ by the definitions of $A_i$ with $1\leq i\leq 7$ and $t_j$ with $1\leq j\leq 4$. This means that each $A_i$ can be represented by $A_{1},A_{2},A_{4}$ and $A_{7}$. Moreover, we always assume that the number of linearly independent elements in the set  $\big\{\lambda^{-1}_1, \lambda^{-1}_2, \lambda^{-1}_3, (\lambda_1+\lambda_3)^{-1}, (\lambda_1+\lambda_2)^{-1}, (\lambda_2+\lambda_3)^{-1}, (\lambda_1+\lambda_2+\lambda_3)^{-1}\big\}$ over $\mathbb{F}_2$ equals $\wt(t_1,t_2,t_3,t_4)+3$. From this assumption, it follows that $m\ge\wt(t_1,t_2,t_3,t_4)+3$. In order to further clarify the possible values of $\wt(t_1,t_2,t_3,t_4)$, the following lemmas are needed.

\begin{lemma}\label{lemma-j}\cite{1967-Berlekamp} Let $m$ be a positive integer and $a,b\in \mathbb{F}_{2^m}$ with $a\ne 0$. Then the quadratic equation $x^2+ax+b=0$ has solutions in $\mathbb{F}_{2^m}$ if and only if $\mathrm{Tr}^m_1(\frac{b}{a^2})=0$.
\end{lemma}
\begin{lemma}\label{lemma-jl} Let $\lambda_i~(1\leq i\leq3)$ be any three linearly independent elements in $\mathbb{F}_{2^m}^{\star}$. Let $\theta$ be a primitive element in $\mathbb{F}_{2^2}$ and $t_1,t_2,t_3,t_4$ be defined as in \textnormal{Definition \ref{zht}}. Then we have $t_1=1$ and $\wt(t_1,t_2,t_3,t_4)\geq3$.
\end{lemma}
\begin{proof} If $t_1=0$, then the equation $(\lambda_1+\lambda_2+\lambda_3)^{-1}=\lambda_1^{-1}+\lambda_2^{-1}+\lambda_3^{-1}$ is equivalent to $\lambda_1^2+\lambda_1(\lambda_2+\lambda_3)+\lambda_2\lambda_3=0$. From this equation, it easily follows that $\lambda_1=\lambda_2$ or $\lambda_1=\lambda_3$, which contradicts the linear independence of $\lambda_i~(1\leq i\leq3)$. Hence $t_1=1$. Next we prove that $\wt(t_1,t_2,t_3,t_4)\geq3$. If $m$ is odd, then it follows directly from Lemma \ref{lemma-j} that $\wt(t_2,t_3,t_4)=3$. If $m$ is even, then for any $u\in\mathbb{F}_{2^m}^{\star}$, the equation $(u+\lambda_i)^{-1}=u^{-1}+\lambda_i^{-1}$ is equivalent to $(\frac{u}{\lambda_i})^2+\frac{u}{\lambda_i}+1=0$. So Lemma \ref{lemma-j} and $1+\theta+\theta^2=0$ yield that the solutions are
\begin{equation}\label{ttlcode_1}
u_{i1}=\lambda_i\theta~\text{and}~u_{i2}=\lambda_i\theta^2.
\end{equation}
For the set $\{t_1,t_2,t_3,t_4\}$, we assert that if any two elements are zero, then $\lambda_i~(1\leq i\leq3)$ are linearly dependent. Without loss of generality, we assume that $t_2=t_3=0$, then by Eq. (\ref{ttlcode_1}) we get that $\lambda_1=\lambda_2\theta(\lambda_2\theta^2)$ and $\lambda_1=\lambda_3\theta(\lambda_3\theta^2)$, which clearly implies that $\lambda_i~(1\leq i\leq3)$ are linearly dependent.
\hfill $\square$
\end{proof}

For the new defined function $f(x)$ in Eq. (\ref{f(x)}), we have the following lemma based on Lemma \ref{lemma-jl}, which will be used in the sequel.
\begin{lemma}\label{t_1}
Let notation be the same as before and $\mathrm{i}\in\{0,1\}$. Let $t_1,t_2,t_3,t_4$ be given in \textnormal{Definition \ref{zht}} and $n=2m$ be an even positive integer with $m\ge \wt(t_1,t_2,t_3,t_4)+3$. Then $\mathcal{W}_{f}(0)=3\cdot2^{n-2}-2^{m-2}$ and the following statements hold.

$(1)$ If $(t_1,t_2,t_3,t_4)=(1,1,1,1),$ then for any $\beta\in\mathbb{F}^{\star}_{2^{n}},$
\begin{eqnarray*}\label{httL}
\mathcal{W}_{f}(\beta)
&=&\left\{ \begin{array}{ll}
(-1)^{\mathrm{i}}2^{m-2}, & \ 35\cdot\left(2^{n-7}+2^{m-7}\right)-\mathrm{i}\left(2^m+1\right) \ \mathrm{times}, \\
(-1)^{\mathrm{i}}3\cdot2^{m-2}, & \ 21\cdot\left(2^{n-7}+2^{m-7}\right) \ \mathrm{times}, \\
(-1)^{\mathrm{i}}5\cdot2^{m-2}, & \ 7\cdot\left(2^{n-7}+2^{m-7}\right) \ \mathrm{times}, \\
(-1)^{\mathrm{i}}7\cdot2^{m-2}, & \ 2^{n-7}+2^{m-7} \ \mathrm{times}.
\end{array} \right.
\end{eqnarray*}

$(2)$ If $t_1=1$ and $\wt(t_2,t_3,t_4)=2$, then for any $\beta\in\mathbb{F}^{\star}_{2^{n}},$
\begin{eqnarray*}\label{3httL}
\mathcal{W}_{f}(\beta)
&=&\left\{ \begin{array}{ll}
7\cdot2^{m-2}, & \ 2^{n-6}+2^{m-6} \ \mathrm{times}, \\
(-1)^{\mathrm{i}}2^{m-2}, & \ (16+3\mathrm{i})\cdot\left(2^{n-6}+2^{m-6}\right)-\mathrm{i}\left(2^m+1\right)\ \mathrm{times}, \\
(-1)^{\mathrm{i}}3\cdot2^{m-2}, & \ (9+3\mathrm{i})\cdot\left(2^{n-6}+2^{m-6)}\right) \ \mathrm{times}, \\
(-1)^{\mathrm{i}}5\cdot2^{m-2}, & \ (4-\mathrm{i})\cdot\left(2^{n-6}+2^{m-6}\right) \ \mathrm{times}.
\end{array} \right.
\end{eqnarray*}
\end{lemma}
\begin{proof} It follows immediately from Eq. (\ref{ht}) that $\mathcal{W}_{f}(0)=3\cdot2^{n-2}-2^{m-2}$. In what follows, we shall calculate the value of $\mathcal{W}_{f}(\beta)$ for any $\beta\in\mathbb{F}^{\star}_{2^{n}}$.

$(1)$ If $(t_1,t_2,t_3,t_4)=(1,1,1,1)$, then by $(-1)^{A_j}=1-2A_j$ with $1\leq j\leq 7$, Eq. (\ref{ht}) can be written as
\begin{eqnarray}\label{htL2}
\mathcal{W}_{f}(\beta)&=&2^{m-2}(2(A_7+A_1+A_2+A_4)-2(A_3+A_5+A_6)-1).
\end{eqnarray}
For the convenience of calculation, we let $\wt(A_7,A_1,A_2,A_4)=\wt_{1}$ and $\wt(A_3,A_5,A_6)=\wt_{2}$. In light of Eq. (\ref{htL2}), we can arrive at the following table.
\begin{table}[H]
\vspace{-0.28cm}
\begin{center}
\caption{~The value of $\mathcal{W}_{f}(\beta)$ when $(t_1,t_2,t_3,t_4)=(1,1,1,1)$~}\label{Ltabel2}
\begin{tabular}{|l|l|l|}
\hline
\makecell[c]{$\mathrm{The}~\mathrm{value}~\mathrm{of}$~$\mathcal{W}_{f}(\beta)$}  &  \makecell[c]{Condition} \\
\hline
$2^{m-2}$  &  $\wt_{1}=j_1, ~\wt_{2}=j_1-1,~ 1\leq j_1\leq 4  $ \\
\hline
$3\cdot2^{m-2}$  &  $\wt_{1}=j_1,~ \wt_{2}=j_1-2,~ 2\leq j_1\leq 4  $      \\
\hline
$5\cdot2^{m-2}$  &  $\wt_{1}=3,~ \wt_{2}=0; ~\wt_{1}=4,~ \wt_{2}=1$  \\
\hline
$7\cdot2^{m-2}$  &  $\wt_{1}=4,~ \wt_{2}=0$ \\
\hline
$-2^{m-2}$  &  $\wt_{1}=\wt_{2}=j_1,~ 0\leq j_1\leq 3  $  \\
\hline
$-3\cdot2^{m-2}$  &  $\wt_{1}=j_1,~ \wt_{2}=j_1+1, ~0\leq j_1\leq 2  $ \\
\hline
$-5\cdot2^{m-2}$  &  $\wt_{1}=0, ~\wt_{2}=2;~ \wt_{1}=1,~ \wt_{2}=3$ \\
\hline
$-7\cdot2^{m-2}$  &  $\wt_{1}=0,~ \wt_{2}=3$ \\
\hline
\end{tabular}
\end{center}
\end{table}
Next we only calculate the number of $\beta$ that satisfies $\mathcal{W}_{f}(\beta)=-7\cdot2^{m-2}$ in $\mathbb{F}^{\star}_{2^{n}}$, as the other cases can be obtained similarly. By using Lemma \ref{function g}, we can easily deduce that
\begin{eqnarray}\label{2htL}
N_{1}&=&\frac{1}{2^7}\sum\limits_{\beta \in \mathbb{F}_{2^n}}\left(1+(-1)^{A_7}\right)\left(1+(-1)^{A_1}\right)\left(1+(-1)^{A_2}\right)\left(1+(-1)^{A_4}\right)\left(1-(-1)^{A_3}\right)\nonumber\\&&\left(1-(-1)^{A_5}\right)\left(1-(-1)^{A_6}\right)\nonumber\\
&=& 2^{n-7}+2^{m-7}.
\end{eqnarray}

Similar to the calculation of Eq. (\ref{2htL}), it follows from Table \ref{Ltabel2} that
\begin{eqnarray*}\label{httL}
\mathcal{W}_{f}(\beta)
&=&\left\{ \begin{array}{ll}
(-1)^{\mathrm{i}}2^{m-2}, & \ 35\cdot\left(2^{n-7}+2^{m-7}\right)-\mathrm{i}\left(2^m+1\right) \ \mathrm{times}, \\
(-1)^{\mathrm{i}}3\cdot2^{m-2}, & \ 21\cdot\left(2^{n-7}+2^{m-7}\right) \ \mathrm{times}, \\
(-1)^{\mathrm{i}}5\cdot2^{m-2}, & \ 7\cdot\left(2^{n-7}+2^{m-7}\right) \ \mathrm{times}, \\
(-1)^{\mathrm{i}}7\cdot2^{m-2}, & \ 2^{n-7}+2^{m-7} \ \mathrm{times}.
\end{array} \right.
\end{eqnarray*}

(2) If $t_1=1$ and $\wt(t_2,t_3,t_4)=2$, then proceeding as in the proof of (1), the desired conclusion follows easily.\hfill $\square$
\end{proof}
\begin{remark}\label{remark1}
According to the condition in Lemma \ref{t_1} that $m\ge \wt(t_1,t_2,t_3,t_4)+3$, we can deduce that $m\ge 6$ in Lemma \ref{t_1}. By the values of $\mathcal{W}_{f}(\beta)$ in Lemma \ref{t_1}, we observe that $\mathcal{W}_{f}(\beta)\pm\mathcal{W}_{f}(0)\equiv 0 \pmod 8$ for all $\beta\in\mathbb{F}_{2^n}$. Furthermore,  using Theorem \ref{lfunction1}, we see that the linear codes $\mathcal{C}^{\star}_{f}$ designed based on Lemma \ref{t_1} are self-orthogonal.
\end{remark}
\begin{lemma}\label{llt}\textnormal{\cite{DHZ18}}
The binary linear code $\mathcal{C}^{\star}_{f}$ in \textnormal{Eq. (\ref{code_1})} has length $2^n-1$ and dimension $n+1.$ In addition, the weight distribution of $\mathcal{C}^{\star}_{f}$ is given by the following multiset union:
$$\left\{\left\{\frac{2^n -\mathcal{W}_{f}(\beta)}{2}~|~\beta\in \mathbb{F}_{2^n}\right\}\cup\left\{2^{n-1}~|~\omega \in \mathbb{F}_{2^n}^{\star}\right\} \cup \{0\}\right\}.$$
Furthermore, $\mathcal{C}^{\star}_{f}$ is a minimal linear code if and only if
$\mathcal{W}_{f}({h})\pm \mathcal{W}_{f}({l})\neq 2^n, \ \mathrm{where} \ h,l\in \mathbb{F}_{2^n} \ \mathrm{and} \ h\neq l.$
\end{lemma}

Using Lemma \ref{t_1} into Lemma \ref{llt}, we get the following result immediately.
\begin{theorem}\label{lt_1}
Let notation and conditions be the same as in \textnormal{Lemma \ref{t_1}}. Let $\mathcal{C}^{\star}_{f}$ be the linear code derived from $f(x)$ in \textnormal{Eq. (\ref{f(x)})}. Then $\mathcal{C}^{\star}_{f}$ is a $\big[2^{n}-1,n+1,2^{n-3}+2^{m-3}\big]$ minimal code with $\frac{w_{\min}}{w_{\max}}<\frac{1}{2}$ and the weight distributions of $\mathcal{C}^{\star}_{f}$ are given by the following \textnormal{Tables \ref{teven-2}} and $\ref{teven-5}$, respectively.
\vspace{0.6cm}
\begin{table}
\begin{center}
\caption{The weight distribution of $\mathcal{C}^{\star}_{f}$ when $(t_1,t_2,t_3,t_4)=(1,1,1,1)$}\label{teven-2}
\begin{tabular}{ll}
\hline\noalign{\smallskip}
Weight  &  Multiplicity   \\
\noalign{\smallskip}
\hline\noalign{\smallskip}
$0$  &  1 \\
$ 2^{n-3}+2^{m-3}$    &  $  1$   \\
$ 2^{n-1}-(-1)^{\mathrm{i}}2^{m-3}$  &  $ 35\cdot\left(2^{n-7}+2^{m-7}\right)-\mathrm{i}\left(2^m+1\right)$     \\
$2^{n-1}-(-1)^{\mathrm{i}}3\cdot2^{m-3}$  &  $21\cdot\left(2^{n-7}+2^{m-7}\right)$    \\
$2^{n-1}-(-1)^{\mathrm{i}}5\cdot2^{m-3}$  &  $7\cdot\left(2^{n-7}+2^{m-7}\right)$    \\
$2^{n-1}-(-1)^{\mathrm{i}}7\cdot2^{m-3}$  &  $2^{n-7}+2^{m-7}$    \\
$2^{n-1}$  &  $2^{n}-1$    \\
\noalign{\smallskip}
\hline
\end{tabular}
\end{center}
\end{table}
\vspace{-0.5cm}
\begin{table}
\begin{center}
\caption{The weight distribution of $\mathcal{C}^{\star}_{f}$ when $t_1=1, \wt(t_2,t_3,t_4)=2$}\label{teven-5}
\begin{tabular}{ll}
\hline\noalign{\smallskip}
Weight  &  Multiplicity   \\
\noalign{\smallskip}
\hline\noalign{\smallskip}
$0$  &  1 \\
$ 2^{n-3}+2^{m-3}$    &  $  1$   \\
$ 2^{n-1}-7\cdot2^{m-3}$  &  $ 2^{n-6}+2^{m-6}$     \\
$2^{n-1}-(-1)^{\mathrm{i}}2^{m-3}$  &  $\left(16+3\mathrm{i}\right)\cdot\left(2^{n-6}+2^{m-6}\right)-\mathrm{i}\left(2^m+1\right)$    \\
$2^{n-1}-(-1)^{\mathrm{i}}3\cdot2^{m-3}$  &  $\left(9+3\mathrm{i}\right)\cdot\left(2^{n-6}+2^{m-6}\right)$    \\
$2^{n-1}-(-1)^{\mathrm{i}}5\cdot2^{m-3}$  &  $\left(4-\mathrm{i}\right)\cdot\left(2^{n-6}+2^{m-6}\right)$    \\
$2^{n-1}$  &  $2^{n}-1$    \\
\noalign{\smallskip}
\hline
\end{tabular}
\end{center}
\end{table}
\vspace{-0.5cm}
\end{theorem}
\begin{example}\label{lexamples}
Let $\lambda_1=\xi^{2^7+1}$, $\lambda_2=\xi^{2^8+2}$ and $\lambda_3=\xi^{2^9+4}$, where $\xi$ is a primitive element in $\mathbb{F}_{2^{14}}$ such that $\xi^{14}+\xi^{7}+\xi^{5}+\xi^{3}+1=0$. Then we have $\wt(t_1,t_2,t_3,t_4)=4$ and the elements $\lambda^{-1}_1, \lambda^{-1}_2, \lambda^{-1}_3, (\lambda_1+\lambda_3)^{-1}, (\lambda_1+\lambda_3)^{-1}, (\lambda_2+\lambda_3)^{-1}, (\lambda_1+\lambda_2+\lambda_3)^{-1}$ are linearly independent over $\mathbb{F}_2$. Furthermore, we can verify by Magma program that the linear code $\mathcal{C}^{\star}_{f}$ has the parameters $[16383, 15, 2064]$ and the weight enumerator is $1+z^{2064}+129z^{8080}+903z^{8112}+2709z^{8144}+4515z^{8176}+16383z^{8192}+4386z^{8208}+2709z^{8240}+903z^{8272}+129z^{8304}$ with $\frac{w_{\min}}{w_{\max}}=\frac{2064}{8304}<\frac{1}{2}$. This is consistent with Theorem \ref{lt_1} (Table \ref{teven-2}).
\end{example}
\begin{example}\label{tlexamples}
Let $\lambda_1=\xi^{2^6+1}$, $\lambda_2=\xi^{(2^6+1)\cdot21}$ and $\lambda_3=1$, where $\xi$ is a primitive element in $\mathbb{F}_{2^{12}}$ such that $\xi^{12}+\xi^{7}+\xi^{6}+\xi^{5}+\xi^{3}+\xi+1=0$. Then we have $t_1=1, \wt(t_2,t_3,t_4)=2$ and the elements $\lambda^{-1}_1, \lambda^{-1}_2, \lambda^{-1}_3, (\lambda_1+\lambda_3)^{-1}, (\lambda_1+\lambda_3)^{-1}, (\lambda_1+\lambda_2+\lambda_3)^{-1}$ are linearly independent over $\mathbb{F}_2$. Furthermore, we can verify by Magma program that the linear code $\mathcal{C}^{\star}_{f}$ has the parameters $[4095, 13, 520]$ and the weight enumerator is $1+z^{520}+65z^{1992}+260z^{2008}+585z^{2024}+1040z^{2040}+4095z^{2048}+1170z^{2056}+780z^{2072}+195z^{2088}$ with $\frac{w_{\min}}{w_{\max}}=\frac{520}{2088}<\frac{1}{2}$. This is consistent with Theorem \ref{lt_1} (Table \ref{teven-5}).
\end{example}

The next result follows by Remark \ref{remark1} or \cite[Theorem 1.4.8]{Huffman-Pless-2003}.

\begin{corollary}\label{p-ary1}
Let notation and conditions be the same as in \textnormal{Theorem \ref{lt_1}}. Then the $\mathcal{C}^{\star}_f$ in \textnormal{Tables \ref{teven-2}} and $\ref{teven-5}$ are self-orthogonal minimal linear codes that violate the AB condition.
\end{corollary}
\subsection{Singly-even self-orthogonal minimal codes violating the AB condition}\label{MM-B1-class}
In this subsection we always assume that $\lambda_1$ and $\lambda_2$ are any two linearly independent elements in $\mathbb{F}_{2^m}^{\star}$ such that $\lambda_1+\lambda_2=\lambda_3$ and $(\lambda_1+\lambda_2)^{-1}\ne \lambda^{-1}_1+\lambda^{-1}_2$. Let $n=2m\geq6$ be an even positive integer. Then for any two distinct elements $w_1,w_2\in\mathbb{F}^{\star}_{2^{n}}$,  we define a new Boolean function as follows.
\begin{eqnarray}\label{thf(x)}
f(x)=\mathrm{Tr}_1^m\left(\lambda_1 x^{2^m+1}\right)\mathrm{Tr}_1^m\left(\lambda_2 x^{2^m+1}\right)+\sum_{i=1}^2(x+w_i)^{2^n-1},x\in\mathbb{F}_{2^{n}},
\end{eqnarray}
where for $i,j\in\{1,2,3\}$, $w_1,w_2$ and $w_3=w_1+w_2$ satisfying
\begin{eqnarray}\label{llthf(x)}
&&\mathrm{Tr}_1^m\left(\lambda_j w_i^{2^m+1}\right)\nonumber\\
&=&\mathrm{Tr}_1^m\left(\left(\lambda^{-1}_1+\lambda^{-1}_2\right)^{-1} w_i^{2^m+1}\right)=\mathrm{Tr}_1^m\left(\left(\lambda^{-1}_1+\lambda^{-1}_3\right)^{-1} w_i^{2^m+1}\right)\nonumber\\
&=&\mathrm{Tr}_1^m\left(\left(\lambda^{-1}_2+\lambda^{-1}_3\right)^{-1} w_i^{2^m+1}\right)=\mathrm{Tr}_1^m\left(\left(\sum_{j=1}^3\lambda^{-1}_j\right)^{-1}w_i^{2^m+1}\right)=0.
\end{eqnarray}
Furthermore, by the Walsh transform of $f(x)$ at $\beta\in\mathbb{F}_{2^n}$ and Lemma \ref{function g}, we have
\begin{eqnarray}\label{cht}
\mathcal{W}_{f}(\beta)&=&\sum\limits_{x\in \mathbb{F}_{2^{n}}\setminus\{w_1,w_2\}}(-1)^{f(x)+\mathrm{Tr}_1^n\left(\beta x\right)}-\left((-1)^{\mathrm{Tr}_1^n\left(\beta w_1\right)}+(-1)^{\mathrm{Tr}_1^n\left(\beta w_2\right)}\right) \nonumber \\
&=&\sum\limits_{x\in \mathbb{F}_{2^{n}}}(-1)^{f(x)+\mathrm{Tr}_1^n\left(\beta x\right)}-2\left((-1)^{\mathrm{Tr}_1^n\left(\beta w_1\right)}+(-1)^{\mathrm{Tr}_1^n\left(\beta w_2\right)}\right) \nonumber \\
&=&\sum\limits_{x \in \mathbb{F}_{2^n}}\sum\limits_{\mathrm{i} \in \mathbb{F}_{2}}(-1)^{\mathrm{Tr}^m_1\left(\mathrm{i}\lambda_1x^{2^m+1}\right)+\mathrm{Tr}_1^n(\beta x)}\frac{1+(-1)^{\mathrm{i}+\mathrm{Tr}^m_1\left(\lambda_2 x^{2^m+1}\right)}}{2}\nonumber \\&&-2\left((-1)^{\mathrm{Tr}_1^n\left(\beta w_1\right)}+(-1)^{\mathrm{Tr}_1^n\left(\beta w_2\right)}\right)\nonumber\\
&=&\frac{1}{2}\sum\limits_{x \in \mathbb{F}_{2^n}}(-1)^{\mathrm{Tr}^m_1(\beta x)}-2^{m-1}\left((-1)^{A'_1}+(-1)^{A'_2}-(-1)^{A'_3}\right)\nonumber \\&&-2\left((-1)^{\mathrm{Tr}_1^n\left(\beta w_1\right)}+(-1)^{\mathrm{Tr}_1^n\left(\beta w_2\right)}\right),
\end{eqnarray}
where $A'_{i}=\mathrm{Tr}_1^m\left(\lambda_{i}^{-1}\beta^{2^m+1}\right)$ with $i\in\{1,2,3\}$.

In the following, for a set $\left\{\delta_1, \delta_2,\cdots,\delta_r\right\}$ of $\mathbb{F}_{2^n},$ we use the symbol $\mathrm{Span}(\delta_1, \delta_2,\\\cdots,\delta_r)$ to denote the $\mathbb{F}_{2}$-subspace spanned by $\left\{\delta_1, \delta_2,\cdots,\delta_r\right\}$.
\begin{definition}\label{defacc1}
Let $V=\mathrm{Span}(\delta_1, \delta_2,\cdots,\delta_{n-2})$, where $\{\delta_1, \delta_2,\cdots,\delta_{n-2}\}$ is a set of linearly independent elements in $\mathbb{F}_{2^n}$ satisfying $\mathrm{Tr}_1^n\left(w_1\delta_j\right) = \mathrm{Tr}_1^n\left(w_2\delta_j\right) = 0$ for all $1 \leq j \leq n-2$.
\end{definition}

According to the above statements, we have the following result.
\begin{lemma}\label{st_1}
Let notation be the same as before and $f(x)$ be defined in \textnormal{Eq. (\ref{thf(x)})}. Let $V$ be defined in \textnormal{Definition \ref{defacc1}} and $\mathrm{i}\in\{0,1\}$. Then for any $\beta\in V$, we have
\begin{eqnarray*}\label{httL}
\mathcal{W}_{f}(\beta)
&=&\left\{ \begin{array}{ll}
2^{n-1}-2^{m-1}-4, & \ 1 \ \mathrm{time}, \\
(-1)^{\mathrm{i}}2^{m-1}-4, & \ 3\cdot2^{n-5}+(3-8\mathrm{i})2^{m-3}-\mathrm{i} \ \mathrm{times}, \\
(-1)^{\mathrm{i}}3\cdot2^{m-1}-4, & \ 2^{n-5}+2^{m-3} \ \mathrm{times}.
\end{array} \right.
\end{eqnarray*}
\end{lemma}
\begin{proof}
For any $\beta\in V$, by Lemma \ref{function g} and Definition \ref{defacc1}, Eq. (\ref{cht}) can be reduced as
\begin{eqnarray}\label{vhta}
\mathcal{W}_{f}(\beta)
&=&\left\{ \begin{array}{ll}
2^{n-1}-2^{m-1}-4, & \mathrm{if}~\beta=0, \\
(-1)^{1+\mathrm{i}}2^{m-1}-4, & \mathrm{if}~\beta\ne0,\left(A'_1, A'_2, A'_3\right)\in\{\left(\mathrm{i},\mathrm{i},\mathrm{i}\right),\\&\left(1+\mathrm{i},\mathrm{i},1+\mathrm{i}\right) ,\left(\mathrm{i},1+\mathrm{i},1+\mathrm{i}\right)\},\\
(-1)^{1+\mathrm{i}}3\cdot2^{m-1}-4, & \mathrm{if}~\beta\ne0,(A'_1, A'_2, A'_3)=(\mathrm{i},\mathrm{i},1+\mathrm{i}).
\end{array} \right.
\end{eqnarray}
Denote by $N$ the number of $\beta$ satisfying $\mathcal{W}_{F}(\beta)=-3\cdot2^{m-1}-4$ in $V$, we now proceed to calculate $N$. Incorporated with Eqs. (\ref{vhta}), (\ref{llthf(x)}) and Lemma \ref{function g}, it follows that
\begin{eqnarray}\label{foc}
N&=&\sum\limits_{\beta \in V}\left(\frac{1+(-1)^{A'_1}}{2}\right)\left(\frac{1+(-1)^{A'_2}}{2}\right)\left(\frac{1-(-1)^{A'_3}}{2}\right)\nonumber\\
&=&\sum\limits_{\beta \in \mathbb{F}_{2^n}}\left(\frac{1+(-1)^{A'_1}}{2}\right)\left(\frac{1+(-1)^{A'_2}}{2}\right)\left(\frac{1-(-1)^{A'_3}}{2}\right)\prod^2_{i=1}\left(\frac{1+(-1)^{\mathrm{Tr}_1^n\left(\beta w_i\right)}}{2}\right)\nonumber\\
&=&2^{n-5}+2^{m-3}\nonumber.
\end{eqnarray}
The frequency of the other cases in Eq. (\ref{vhta}) can be obtained by the method analogous to that used above. \hfill $\square$
\end{proof}
\begin{remark}\label{lremark1} (1) Since we focus on constructing self-orthogonal minimal linear codes  violating the AB condition in this section, we observe that if $(\lambda_1+\lambda_2)^{-1}=\lambda^{-1}_1+\lambda^{-1}_2$, then proceeding as in the proof of Lemma \ref{st_1}, again by Lemma \ref{llt}, the self-orthogonal minimal linear codes satisfying the AB condition can be obtained. Moreover, if $m=2$, then $(\lambda_1+\lambda_2)^{-1}=\lambda^{-1}_1+\lambda^{-1}_2$ holds. Therefore, throughout this subsection, we always assume $(\lambda_1+\lambda_2)^{-1}\ne\lambda^{-1}_1+\lambda^{-1}_2$ and $n=2m\geq6$.

(2) In fact, we note that $w_i~(1\leq i\leq 3)$ only need to satisfy $\mathrm{Tr}_1^m\left(\lambda_1 w_i^{2^m+1}\right)=0$ or $\mathrm{Tr}_1^m\left(\lambda_2 w_i^{2^m+1}\right)=0$ in Eq. (\ref{llthf(x)}), however, variations in the other terms of  Eq. (\ref{llthf(x)}) will lead to changes for the Walsh transform coefficient frequencies in Eq. (\ref{vhta}). Hence we  consider $w_i~(1\leq i\leq 3)$ satisfy Eq. (\ref{llthf(x)}) without loss of generality.

(3) By the values of $\mathcal{W}_{f}(\beta)$ in Lemma \ref{t_1}, we observe that $\mathcal{W}_{f}(\beta)\pm\mathcal{W}_{f}(0)\equiv 0 \pmod 8$ for all $\beta\in V$. Furthermore,  using Theorem \ref{lfunction1}, we see that the linear codes $\mathcal{C}^{\star}_{f}$ designed based on Lemma \ref{st_1} are self-orthogonal.
\end{remark}

Using Lemma \ref{st_1} into Lemma \ref{llt}, and by Remark \ref{lremark1}(3), we get the following result immediately.
\begin{theorem}\label{ltt_1}
Let notation and conditions be the same as in \textnormal{Lemma \ref{st_1}}. Let $\mathcal{C}^{\star}_{f}$ be the linear code derived from $f(x)$ in \textnormal{Eq. (\ref{thf(x)})}. Then $\mathcal{C}^{\star}_{f}$ is a $\big[2^{n}-1,n-1, 2^{n-2}+2^{m-2}+2\big]$ self-orthogonal minimal code with $\frac{w_{\min}}{w_{\max}}<\frac{1}{2}$ and the weight distribution of $\mathcal{C}^{\star}_{f}$ is given by the following \textnormal{Table \ref{teven-l}.
\vspace{-0.2cm}
\begin{table}[H]
\begin{center}
\caption{The weight distribution of $\mathcal{C}^{\star}_{f}$ }\label{teven-l}
\begin{tabular}{ll}
\hline\noalign{\smallskip}
Weight  &  Multiplicity   \\
\noalign{\smallskip}
\hline\noalign{\smallskip}
$0$  &  1 \\
$ 2^{n-2}+2^{m-2}+2$    &  $  1$   \\
$ 2^{n-1}-(-1)^{\mathrm{i}}2^{m-2}+2$  &  $ 3\cdot2^{n-5}+(3-8\mathrm{i})2^{m-3}-\mathrm{i}$     \\
$2^{n-1}-(-1)^{\mathrm{i}}3\cdot2^{m-2}+2$  &  $ 2^{n-5}+2^{m-3}$    \\
$2^{n-1}$  &  $2^{n-2}-1$    \\
\noalign{\smallskip}
\hline
\end{tabular}
\end{center}
\end{table}}
\end{theorem}
\begin{example}\label{tlexamples}
Let $\lambda_1=\xi^{2^8+1}$, $\lambda_2=\xi^{2^9+2}$, $w_1=\xi^{(2^8+1)\cdot12}$ and $w_2=\xi^{(2^8+1)\cdot164}$, where $\xi$ is a primitive element in $\mathbb{F}_{2^{16}}$ such that $\xi^{16}+\xi^{5}+\xi^{3}+\xi^{2}+1=0$. Then we have $(\lambda_1+\lambda_2)^{-1}\ne\lambda^{-1}_1+\lambda^{-1}_2$ and the elements $w_1,w_2$ satisfy Eq. (\ref{llthf(x)}). Furthermore, we can verify by Magma program that the linear code $\mathcal{C}^{\star}_{f}$ has the parameters $[65535, 15, 16450]$ and the weight enumerator is $1+z^{16450}+2080z^{32578}+6240z^{32706}+16383z^{32768}+5983z^{32834}+2080z^{32962}$ with $\frac{w_{\min}}{w_{\max}}=\frac{16450}{32962}<\frac{1}{2}$. This is consistent with Theorem \ref{ltt_1}.
\end{example}
\begin{remark}\label{remark2}
(1) In \cite[Theorem 1.4.8]{Huffman-Pless-2003}, the authors provided a sufficient condition for determining whether a binary linear code $\mathcal{C}$ is self-orthogonal, that is, the code $\mathcal{C}$ is doubly-even. However, it fails to determine whether a singly-even code $\mathcal{C}$ is self-orthogonal. In this section, we have partially solved the aforementioned issue through Theorem \ref{lfunction1}. To be specific, it follows obviously from Table \ref{teven-l} that the code $\mathcal{C}^{\star}_{f}$ in Theorem \ref{ltt_1} is singly-even; therefore, \cite[Theorem 1.4.8]{Huffman-Pless-2003} cannot be used to determine whether it is self-orthogonal. However, it follows from Theorem \ref{lfunction1} and Lemma \ref{st_1} that the code $\mathcal{C}^{\star}_{f}$ in Theorem \ref{ltt_1} is self-orthogonal.

(2) We compared our work with a substantial body of literature on binary linear codes, while it is not possible to list all the references compared here. Therefore, we have selected only some minimal codes violating the AB condition, such as \cite{Bonini-2021,Chang-2018,DHZ18,Qu-Li-Jin-2025,Li-Yue-2020,Li-2023,2020-Mesnager-Qi-Ru-Tang,Pasalic-2021,Tao-Feng-2021,2020-Xu-Qu,Zhang-Yan-2021,Zhang-2021,Zhang-2022}, for discussion. Through comparison, we found that the weight distributions in those references differ from these of our constructed codes in Theorems \ref{lt_1} and \ref{ltt_1}.
\end{remark}
\section{Self-orthogonal $p$-ary linear codes}\label{M-B-class}
In this section, we always assume that $p$ is an odd prime and that $n$ is a positive integer. First, a necessary and sufficient condition for a $p$-ary linear code to be self-orthogonal shall be acquired by the Walsh coefficient of a $p$-ary function. Then, we shall design several classes of linear codes based on two classes of $p$-ary functions, which contains one class of optimal codes and two classes of minimal codes violating the AB condition. Using the acquired conditions, we shall determine in which cases these codes are self-orthogonal and under which cases they are not. In addition, the following two lemmas are needed in the subsequent proofs.
\begin{lemma}\textnormal{\cite{Huffman-Pless-2003}}\label{3-ary}
$\mathcal{C}$ is a ternary self-orthogonal code over $\mathbb{F}_{3}$ if and only if the weight of each codeword in $\mathcal{C}$ is divisible by three.
\end{lemma}
\begin{lemma}\textnormal{\cite{Wan-1998}}\label{p-ary}
Let $p$ be an odd prime and $\mathcal{C}$ be a $p$-ary linear code. Then $\mathcal{C}$ is self-orthogonal if and only if $\mathbf{c}\cdot\mathbf{c}=0$ for all $\mathbf{c}\in \mathcal{C}$.
\end{lemma}

Below we present an important result of this section, which will be frequently used to determine the self-orthogonality of the codes constructed subsequently.
\begin{theorem}\label{bent function1}
Let notation be the same as before and $p>3$. Let $S=\Big\{\overline{\beta}=(\beta_0,\beta_1,\cdots,\beta_{p-1})\in\mathbb{Z}\ | \  \mathcal{W}_{f}(\beta)=\sum_{a\in\mathbb{F}_{p}}\beta_a\zeta_{p}^{a}, \beta\in\mathbb{F}_{p^n}\Big\}$. Let $\mathcal{C}_{f}$ and $\overline{\mathcal{C}_{f}}$ be defined in \textnormal{Eqs. (\ref{vzcode_1})} and \textnormal{(\ref{lcode_1})}, respectively. Then the following statements hold.

$(1)$ The code $\mathcal{C}_{f}$ is self-orthogonal if and only if $\sum_{a\in\mathbb{F}_{p}}\beta_{a}a^2\equiv 0 \pmod p$ for all $\overline{\beta}\in S$.

$(2)$ The code $\overline{\mathcal{C}_{f}}$ is self-orthogonal if and only if $\sum_{a\in\mathbb{F}_{p}}\beta_{a}a^2\equiv 0 \pmod p$ and $\sum_{a\in\mathbb{F}_{p}}\beta_{a}a\equiv 0 \pmod p$ for all $\overline{\beta}\in S$.
\end{theorem}
\begin{proof}Assume that $\alpha=0$ and $(c, \beta) \in \mathbb{F}_{p}\times\mathbb{F}_{p^n}$. Then we have codeword $\mathbf{c}_{0,\beta, c}=(-c-\mathrm{Tr}^n_1(\beta x))_{x\in\mathbb{F}^{\star}_{p^n}}\in\overline{\mathcal{C}_{f}}$. The definition of $\mathrm{Tr}^n_1(\beta x)$ yields that
\begin{eqnarray}\label{azero}
\hspace{-0.5cm}\mathbf{c}_{0,\beta,c}\cdot\mathbf{c}_{0,\beta,c}&=&\sum\limits_{x\in \mathbb{F}_{p^n}}\left(c+\mathrm{Tr}^n_1(\beta x)\right)^2=\sum\limits_{x\in \mathbb{F}_{p^n}}\left(c^2+2c\mathrm{Tr}^n_1(\beta x)+\mathrm{Tr}^n_1(\beta x)^2\right)\nonumber\\
&=&p^{n}c^2+2cp^{n-1}\sum\limits_{t\in \mathbb{F}_{p}}\beta^{p^n-1} t+p^{n-1}\sum\limits_{t\in \mathbb{F}_{p}}\beta^{p^n-1} t^2\equiv0\hspace{-0.18cm}\pmod p.
\end{eqnarray}
Next, let $N_a$ denote the frequency of each codeword $\mathbf{c}_{\alpha,\beta, c}=(c_0,c_1,\cdots,c_{p^n-1})\in \overline{\mathcal{C}_{f}}$ with component $a\in\mathbb{F}^{\star}_{p}$. Then for any $(\alpha, \beta, c) \in \mathbb{F}^{\star}_{p}\times\mathbb{F}_{p^n}\times\mathbb{F}_{p}$, it follows from Eq. (\ref{lcode_1}) that
\begin{eqnarray}\label{hatf{c}}
N_a&=&\frac{1}{p}\sum\limits_{z\in \mathbb{F}_{p}}\sum\limits_{x\in \mathbb{F}_{p^n}}\zeta_{p}^{\left(\alpha f(x)-\mathrm{Tr}^n_1(\beta x)-c-a\right)z}\nonumber \\
&=&p^{n-1}+\frac{1}{p}\sum\limits_{z\in \mathbb{F}^{\star}_{p}}\sigma_{z\alpha}\left(\mathcal{W}_{f}(\alpha^{-1}\beta)\zeta_{p}^{-(c+a)\alpha^{-1}}\right)\nonumber\\
&=&p^{n-1}+\frac{1}{p}\sum\limits_{z\in \mathbb{F}^{\star}_{p}}\left(\sum\limits_{t\in \mathbb{F}_{p}}\beta_t\zeta_{p}^{(t\alpha-(c+a))z}\right)=p^{n-1}+\frac{1}{p}\sum\limits_{t\in\mathbb{F}_{p}}\left(\sum\limits_{z\in \mathbb{F}^{\star}_{p}}\beta_t\zeta_{p}^{(t\alpha-(c+a))z}\right)\nonumber\\
&=&p^{n-1}-\frac{1}{p}\sum\limits_{t\in\mathbb{F}_{p}}\beta_t+\beta_{(c+a)\alpha^{-1}}.
\end{eqnarray}

Note that $N_a$ is an integer, it yields that $p$ divides $\sum_{t\in\mathbb{F}_{p}}\beta_t$.
Then for any $(\alpha, \beta, c) \in \mathbb{F}^{\star}_{p}\times\mathbb{F}_{p^n}\times\mathbb{F}_{p}$, in light of Eq. (\ref{hatf{c}}) and $\sum_{a\in\mathbb{F}_{p}}a^2=2\frac{w^{2\frac{p-1}{2}}-1}{w^2-1}=2(1+w^2+\cdots+w^{2(\frac{p-1}{2}-1)}\equiv 0 \pmod p$, where $w$ is a primitive element in $\mathbb{F}_{p}$, it follows that
\begin{eqnarray}\label{ha}
\mathbf{c}_{\alpha,\beta,c}\cdot\mathbf{c}_{\alpha,\beta,c}&=&\sum\limits_{a\in\mathbb{F}_{p}}N_aa^2=p^{n-1}\sum\limits_{a\in\mathbb{F}_{p}}a^2-\frac{1}{p}\sum\limits_{t\in\mathbb{F}_{p}}\beta_t\sum\limits_{a\in\mathbb{F}_{p}}a^2
+\sum\limits_{a\in\mathbb{F}_{p}}\beta_{(c+a)\alpha^{-1}}a^2\nonumber\\
&=&\sum\limits_{a\in\mathbb{F}_{p}}\beta_{(c+a)\alpha^{-1}}a^2\equiv\sum\limits_{a\in\mathbb{F}_{p}}\beta_{a}(\alpha a-c)^2\nonumber\\
&=&\sum\limits_{a\in\mathbb{F}_{p}}\left(\alpha^{2}\beta_{a}a^2-2c\alpha\beta_{a}a+c^{2}\beta_{a}\right)\nonumber\\
&=&\alpha^{2}\left(\sum\limits_{a\in\mathbb{F}_{p}}\beta_{a}a^2\right)-2c\alpha\left(\sum\limits_{a\in\mathbb{F}_{p}}\beta_{a}a\right)
+c^{2}\sum\limits_{a\in\mathbb{F}_{p}}\beta_{a}\pmod p.
\end{eqnarray}
Combining Eq. (\ref{azero}), Eq. (\ref{ha}) and Lemma \ref{p-ary}, we can deduce that $\overline{\mathcal{C}_{f}}$ is self-orthogonal if and only if
\begin{align}\label{lha}
\alpha^{2}\left(\sum_{a\in\mathbb{F}_{p}}\beta_{a}a^2\right)-2c\alpha\left(\sum_{a\in\mathbb{F}_{p}}\beta_{a}a\right)+c^{2}\sum\limits_{a\in\mathbb{F}_{p}}\beta_{a}\equiv 0\pmod p,~\mathrm{for ~all}~\overline{\beta}\in \mathrm{S}.
\end{align}

Finally, the first point can be obtained immediately from Eq. (\ref{lha}) with $c=0$; the second point can be derived easily from Eq. (\ref{lha}) and $p$ divides $\sum_{a\in\mathbb{F}_{p}}\beta_a$, and the fact that $\mathcal{C}_{f}$ is a subcode of $\overline{\mathcal{C}_{f}}$.   \hfill $\square$
\end{proof}
\begin{remark}\label{remark2} (1) Note that if $p=3$, then $\sum_{a\in\mathbb{F}_{3}}a^2\equiv 2 \pmod 3$, which fails to satisfy the key point $\sum_{a\in\mathbb{F}_{p}}a^2\equiv 0 \pmod p$ in the proof of Theorem \ref{bent function1}. Consequently, the conclusions of Theorem \ref{bent function1} do not cover the case $p=3$.

(2) It is evident that although the scope of application for \cite[Theorem 3.2]{Heng-Li2024} is broad, it has many restrictions. Since it cannot be used to determine whether a linear code that does not contain all-$1$ vector or is not $p$-divisible is self-orthogonal. However, Theorem \ref{bent function1} can clearly determine whether the codes obtained based on Construction 1 are self-orthogonal without the restrictions of having all-$1$ vectors and $p$-divisible.
\end{remark}
\subsection{Optimal (almost optimal) self-orthogonal linear codes with respect to the Griesmer bound (Singleton bound) from special $p$-ary functions }\label{l-M-B-class}
In this subsection, we first consider linear codes derived from $p$-ary functions. Then, using Theorem \ref{bent function1}, we shall determine under what conditions the constructed codes are self-orthogonal or not. To achieve this, the following two lemmas are necessary.

\begin{lemma}\label{lemma-15}
Let $n$ be a positive integer and $t\in \mathbb{F}^{\star}_{p}$. Let $f(x)$ be a $p$-ary function defined on $\mathbb{F}_{p^n}$ and given by
\begin{eqnarray}\label{g1}
f_{t}(x)
&=&\left\{ \begin{array}{ll}
t, &\mathrm{if}~ x=0, \\
x^{\frac{p^n-1}{2}}, &\mathrm{if}~x\in\mathbb{F}^{\star}_{p^n}.
\end{array} \right.
\end{eqnarray} Then the Walsh transform coefficient of $f_{t}(x)$ at $\beta\in \mathbb{F}_{p^n}$ is equal to
\begin{eqnarray}\label{lta3}
\mathcal{W}_{f_{t}}(\beta)
&=&\left\{ \begin{array}{ll}
\frac{p^n-1}{2}\left(\zeta_p+\zeta^{-1}_p\right)+\zeta^{t}_p, &\mathrm{if}~ \beta=0, \\
\frac{(-1)^{n-1}\sqrt{p^{*}}^{n}-1}{2}\zeta_p-\frac{(-1)^{n-1}\sqrt{p^{*}}^{n}+1}{2}\zeta^{-1}_p+\zeta^{t}_p, &\mathrm{if}~\eta(-\beta)=1,\\
\frac{(-1)^{n}\sqrt{p^{*}}^{n}-1}{2}\zeta_p-\frac{(-1)^{n}\sqrt{p^{*}}^{n}+1}{2}\zeta^{-1}_p+\zeta^{t}_p, &\mathrm{if}~\eta(-\beta)=-1.
\end{array} \right.
\end{eqnarray}
\end{lemma}
\begin{proof} By the definition of Walsh transform of $p$-ary functions and Lemma \ref{lequations}, we have

\begin{eqnarray*}\label{ct}
\mathcal{W}_{f_{t}}(\beta)
&=&\sum\limits_{x\in \mathbb{F}_{p^n}}\zeta_p^{f(x)-\mathrm{Tr}^n_1(\beta x)}
=\zeta^{t}_p+\sum\limits_{x\in \mathbb{F}^{\star}_{p^n}}\zeta_p^{x^{\frac{p^n-1}{2}}-\mathrm{Tr}^n_1(\beta x)}\nonumber\\
&=&\zeta^{t}_p+\sum\limits_{\eta(x)=1}\zeta_p^{1-\mathrm{Tr}^n_1(\beta x)}+\sum\limits_{\eta(x)=-1}\zeta_p^{-1-\mathrm{Tr}^n_1(\beta x)}\nonumber\\
&=&\zeta^{t}_p+\sum\limits_{x\in \mathbb{F}_{p^n}}\left(\zeta_p^{1-\mathrm{Tr}^n_1(\beta x)}\frac{1+\eta(x)}{2}+\zeta_p^{-1-\mathrm{Tr}^n_1(\beta x)}\frac{1-\eta(x)}{2}\right)-\frac{1}{2}\zeta_p-\frac{1}{2}\zeta_p^{-1}\nonumber\\
&=&\zeta_p-\frac{1}{2}\left(\zeta_p+\zeta_p^{-1}\right)+\frac{1}{2}\sum\limits_{x\in \mathbb{F}_{p^n}}\zeta_p^{-\mathrm{Tr}^n_1(\beta x)}\left(\zeta_p+\zeta_p^{-1}\right)\nonumber\\&&+\frac{1}{2}\sum\limits_{x\in \mathbb{F}_{p^n}}\eta(x)\zeta_p^{-\mathrm{Tr}^n_1(\beta x)}\left(\zeta_p-\zeta_p^{-1}\right)\nonumber\\
&=&\left\{ \begin{array}{ll}
\frac{p^n-1}{2}\left(\zeta_p+\zeta^{-1}_p\right)+\zeta^{t}_p, &\mathrm{if}~ \beta=0, \\
\frac{(-1)^{n-1}\eta(-\beta)\sqrt{p^{*}}^{n}-1}{2}\zeta_p-\frac{(-1)^{n-1}\eta(-\beta)\sqrt{p^{*}}^{n}+1}{2}\zeta^{-1}_p+\zeta^{t}_p, &\mathrm{if}~\beta\ne 0.
\end{array} \right.
\end{eqnarray*}
This completes the proof. \hfill $\square$
\end{proof}

Let $S=\left\{\overline{\beta}=(\beta_0,\beta_1,\cdots,\beta_{p-1})\in\mathbb{Z}\ | \ \mathcal{W}_{f_t}(\beta)=\sum_{a\in\mathbb{F}_{p}}\beta_a\zeta_{p}^{a}, \beta\in\mathbb{F}_{p^n}\right\}$. Then by using Lemma \ref{lemma-15}, we see that there exists $\overline{\beta}\in S$ such that $\sum_{a\in\mathbb{F}_{p}}\beta_{a}a\ne 0 \pmod p$. This means that the linear code $\overline{\mathcal{C}_{f_{t}}}$ designed based on the function in Eq. (\ref{g1}) is not self-orthogonal by Theorem \ref{bent function1}. Applying Lemma \ref{lemma-15} to Lemma \ref{1code-lemma}, a simple calculation reveals that $\overline{\mathcal{C}_{f_{t}}}$ is not $p$-divisible. Thus, we also can not determine whether $\overline{\mathcal{C}_{f_{t}}}$ is self-orthogonal through \cite[Theorem 3.2]{Heng-Li2024}. Since our focus is on self-orthogonal linear codes, it suffices to provide the weight distribution of $\mathcal{C}_{f_{t}}$.

To determine under what conditions the $\mathcal{C}_{f_{t}}$ is self-orthogonal, the following lemma is needed.
\begin{lemma}\label{tlemma-8}
Let notation be the same as in \textnormal{Lemma \ref{lemma-15}} and $p>3$. Let $S=\Big\{\overline{\beta}=(\beta_0,\beta_1,\cdots, \beta_{p-1})\in\mathbb{Z}\ | \ \mathcal{W}_{f_{t}}(\alpha^{-1}\beta)=\sum_{z\in\mathbb{F}_{p}}\beta_z\zeta_{p}^{z}, (\alpha,\beta) \in\mathbb{F}^{\star}_{p}\times\mathbb{F}_{p^n}\Big\}$. Then $\sum_{z\in\mathbb{F}_{p}}\beta_{z}z^2\equiv 0 \pmod p$ for all $\overline{\beta}\in S$ if and only if $t\in\{1,-1\}$.
\end{lemma}
\begin{proof}
If $n$ is even, then it follows immediately from Eq. (\ref{lta3}) that $\sum_{z\in\mathbb{F}_{p}}\beta_{z}z^2\equiv -1+t^2 \pmod p$ for all $\overline{\beta}\in S$.

If $n$ is odd, then from Eq. (\ref{lta3}) it is easy to show that $\sum_{z\in\mathbb{F}_{p}}0_{z}z^2\equiv -1+t^2 \pmod p$ when $\beta=0$. For $\beta\in \mathbb{F}_{p^n}^{\star}$, since $\sum_{z\in\mathbb{F}_{p}^{\star}}\eta_0(z)\zeta_{p}^{z}=\sqrt{p^{*}}$, Eq. (\ref{lta3}) can be written as
\begin{eqnarray}\label{ta3}
\mathcal{W}_{f_{t}}(\beta)
&=&\frac{\eta(-\beta)p^{*\frac{n-1}{2}}}{2}\left(\sum\limits_{z\in\mathbb{F}_{p}^{\star}}\eta_0(z)\zeta_{p}^{z+1}-\sum\limits_{z\in\mathbb{F}_{p}^{\star}}\eta_0(z)\zeta_{p}^{z-1}
\right)-\frac{1}{2}\left(\zeta_p+\zeta_p^{-1}\right)+\zeta^{t}_p\nonumber\\
&=&\frac{\eta(-\beta)p^{*\frac{n-1}{2}}}{2}\left(\sum\limits^{p-1}_{z=3}(\eta_0(z-2)-\eta_0(z))\zeta_{p}^{z-1}+\eta_0(-1)-\eta_0(1)\right)\nonumber\\&&
+\frac{\eta(-\beta)p^{*\frac{n-1}{2}}\eta_0(-2)-1}{2}\zeta_p^{-1}-\frac{\eta(-\beta)p^{*\frac{n-1}{2}}\eta_0(2)+1}{2}\zeta_p+\zeta^{t}_p.
\end{eqnarray}

Below, Eq. (\ref{ta3}) yields that
\begin{eqnarray}\label{tta3}
\sum\limits_{a\in\mathbb{F}_{p}}\beta_{z}z^2
&=&\frac{\eta(-\beta)p^{*\frac{n-1}{2}}}{2}\left(\sum\limits^{p-1}_{z=3}\left(\eta_0(z-2)-\eta_0(z)\right)(z-1)^2\right)\nonumber\\&&+\frac{\eta(-\beta)p^{*\frac{n-1}{2}}}{2}(\eta_0(-2)-\eta_0(2))-1+t^2\nonumber\\
&\equiv&\left\{ \begin{array}{ll}
-1+t^2\pmod p, &\mathrm{if}~ n\ge3, \\
\frac{\eta(-\beta)}{2}\sum\limits^{p-1}_{z=3}(\eta_0(z-2)-\eta_0(z))(z-1)^2\\+\frac{\eta(-\beta)}{2}(\eta_0(-2)-\eta_0(2))-1+t^2\pmod p, &\mathrm{if}~n=1.
\end{array} \right.
\end{eqnarray}

For $n=1$, Eq. (\ref{tta3}) can be written as
\begin{eqnarray*}\label{ttla3}
\sum\limits_{a\in\mathbb{F}_{p}}\beta_{z}z^2
&=&\frac{\eta(-\beta)\sum_{z\in \mathbb{F}^{\star}_{p}}\eta_0(z-2)\left((z-2)^2+2(z-2)+1\right)-\eta(-\beta)\eta_0(2)}{2}\nonumber\\&&-\frac{\eta(-\beta)\sum_{z\in \mathbb{F}^{\star}_{p}}\eta_0(z)\left(z^2-2z+1\right)}{2}+\frac{\eta(-\beta)(\eta_0(-2)-\eta_0(2))}{2}-1+t^2\nonumber\\
&=&\frac{-\eta(-\beta)\left(\eta_0(-2)-\eta_0(2)\right)}{2}+\frac{\eta(-\beta)\left(\eta_0(-2)-\eta_0(2)\right)}{2}-1+t^2\nonumber\\
&=&-1+t^2\pmod p,
\end{eqnarray*}
where the second equation holds by $\sum_{z\in \mathbb{F}^{\star}_{p}}\eta_0(z)z^2=\sum_{z\in \mathbb{F}^{\star}_{p}}\eta_0(z)z\equiv0\pmod p$.

Based on the above analysis, the desired results follow immediately. \hfill $\square$
\end{proof}

With the help of the above arguments, we now proceed to prove the following result.
\begin{theorem}\label{ccwwl}
Let $n$ be a positive integer and $\mathcal{C}_{f_{t}}$ be a linear code derived from the $f_{t}(x)$ in \textnormal{Eq. (\ref{g1})}. Then the following statements hold.

$(1)$ If $p^{n}\equiv1\pmod4$, then $\mathcal{C}_{f_{t}}$ is a $\left[p^{n},n+1, p^{n}-p^{n-1}\right]$ code with weight distribution in \textnormal{Table \ref{tteven}}. Moreover, $\mathcal{C}_{f_{t}}$ is self-orthogonal if and only if $t\in\{1,-1\}$.

$(2)$ If $p^{n}\equiv3\pmod4$ with $p>3$, then $\mathcal{C}_{f_{t}}$ is a $\left[p^{n},n+1, (p-1)p^{n-1}-p^{\frac{n-1}{2}}\right]$ code with weight distribution in \textnormal{Table \ref{tleven}}. Furthermore, $\mathcal{C}_{f_{t}}$ is self-orthogonal if and only if $t\in\{1,-1\}$. For $p=3$, the code $\mathcal{C}_{f_{t}}$ is self-orthogonal if and only if $n>2$.
\begin{table}[H]
\begin{center}
\caption{The weight distribution of $\mathcal{C}_{f_{t}}$ when $p^{n}\equiv1\pmod4$}\label{tteven}
\begin{tabular}{ll}
\hline\noalign{\smallskip}
Weight  &  Multiplicity   \\
\noalign{\smallskip}
\hline\noalign{\smallskip}
$0$  &  1 \\
$ p^{n}$    &  $  p-1$   \\
$ (p-1)p^{n-1}$  &  $ p^{n+1}-p$     \\
\noalign{\smallskip}
\hline
\end{tabular}
\end{center}
\end{table}
\vspace{-0.5cm}
\begin{table}[H]
\begin{center}
\caption{The weight distribution of $\mathcal{C}_{f_{t}}$ when $p^{n}\equiv3\pmod4$}\label{tleven}
\begin{tabular}{ll}
\hline\noalign{\smallskip}
Weight  &  Multiplicity   \\
\noalign{\smallskip}
\hline\noalign{\smallskip}
$0$  &  1 \\
$ p^{n}$    &  $  p-1$   \\
$ (p-1)p^{n-1}$  &  $ p^{n}-1$     \\
$(p-1)p^{n-1}-p^{\frac{n-1}{2}}$  &  $(p-1)\frac{p^n-1}{2}$    \\
$(p-1)p^{n-1}+p^{\frac{n-1}{2}}$  &  $(p-1)\frac{p^n-1}{2}$  \\
\noalign{\smallskip}
\hline
\end{tabular}
\end{center}
\end{table}
\vspace{-0.5cm}
\end{theorem}
\begin{proof} By substituting Eq. (\ref{lta3}) into Eq. (\ref{codea3}), and then using $\sigma_z\left(\sqrt {p^*}\right)=\eta_0(z)\sqrt{p^*}$ with $z\in\mathbb{F}^{\star}_{p}$ and Lemma \ref{equations}, we can obtain Tables \ref{tteven} and \ref{tleven} after simple calculations. Hence the desired results follow by combining Theorem \ref{bent function1}, Lemmas \ref{3-ary} and \ref{tlemma-8}.   \hfill $\square$
\end{proof}
\begin{example}\label{lexamples linear code}
(1) Let $p=3$, $n=4$, $t=1$ and $\xi$ be a primitive element in $\mathbb{F}_{3^4}$ such that $\xi^4+2\xi^3+2=0$. Then we can verify by Magma program that the linear code $\mathcal{C}_{f_t}$ has the parameters $[243, 6, 54]$ and the weight enumerator is $1+240z^{54}+2z^{81}$. Furthermore, let $p=5$, $n\in\{1,2\}$ and $t\in\{1,-1\}$. Then we can verify by Magma program that the linear code $\mathcal{C}_{f_t}$ is indeed self-orthogonal. However, for $t\in\mathbb{F}^{\star}_{5}\setminus\{1,-1\}$, the code $\mathcal{C}_{f_t}$ is not self-orthogonal. This is consistent with Theorem \ref{ccwwl}(1).

(2) Let $p=3$, $n=5$, $t=1$ and $\xi$ be a primitive element in $\mathbb{F}_{3^5}$ such that $\xi^5+2\xi+1=0$. Then we can verify by Magma program that the linear code $\mathcal{C}_{f_t}$ has the parameters $[243, 6, 153]$ and the weight enumerator is $1+242z^{153}+242z^{162}+242z^{171}+2z^{243}$. Furthermore, let $p=7$, $n\in\{1,3\}$ and $t\in\{1,-1\}$. Then we can verify by Magma program that the linear code $\mathcal{C}_{f_t}$ is indeed self-orthogonal. However, for $t\in\mathbb{F}^{\star}_{7}\setminus\{1,-1\}$, the code $\mathcal{C}_{f_t}$ is not self-orthogonal. This is consistent with Theorem \ref{ccwwl}(2).
\end{example}
\begin{remark}\label{ttremark} (1) It is obvious that the code $\mathcal{C}_{f_{t}}$ in Theorem \ref{ccwwl} is not self-orthogonal for $t \in\mathbb{F}^{\star}_{p}\backslash\{1,-1\}$. However, \cite[Theorem 3.2]{Heng-Li2024} can not be used to determine under what conditions $\mathcal{C}_{f_{t}}$ is self-orthogonal. Moreover, it needs to be noted that, according to Lemma \ref{2bt}, the code $\mathcal{C}_{f_{t}}$ in Theorem \ref{ccwwl}$(1)$ is optimal with respect to the Griesmer bound. Particularly, in light of Lemma \ref{2b}, the code $\mathcal{C}_{f_{t}}$ in Theorem \ref{ccwwl}$(1)$ (resp., Theorem \ref{ccwwl}$(2)$) is optimal (resp., almost optimal) with respect to the Singleton bound for $n=1$.

(2) Since Theorem \ref{bent function1} is derived from Lemma \ref{p-ary}, it is also feasible to directly apply Lemma \ref{p-ary} to determine whether the code in Theorem \ref{ccwwl} is self-orthogonal. However, it would still require going through the proof process of Theorem \ref{bent function1}. Therefore, using Theorem \ref{bent function1} directly to verify the self-orthogonality of the code in Theorem \ref{ccwwl} is more straightforward and simpler.
\end{remark}
\subsection{Self-orthogonal minimal linear codes violating the AB condition from the unbalanced plateaued functions}\label{p-p-class}
In this subsection, let $n$ and $s$ be positive integers and $0<s<n$. We shall apply the unbalanced weakly regular plateaued functions to construct several classes of minimal linear codes violating the AB condition and determine their weight distributions completely. Then, we shall investigate the self-orthogonality of the constructed codes by Theorem \ref{bent function1}. To this end, we let
\begin{eqnarray}\label{pfunction}
f_{a}(x)=f(x)+af(x)^{p-1},
\end{eqnarray}
where $f(x)$ is a unbalanced weakly regular plateaued function from $\mathbb{F}_{p^{n}}$ to $\mathbb{F}_{p}$ and $a\in\mathbb{F}^{\star}_{p}$.
\begin{lemma}\label{lemma-8}
Let $f_{a}(x)$ be the function defined in \textnormal{Eq. (\ref{pfunction})} and $\epsilon\in\{1,-1\}$.
Then the Walsh transform coefficient of $f_{a}(x)$ at $\beta\in \mathbb{F}_{p^n}$ is
\begin{eqnarray*}\label{ca3}
\mathcal{W}_{f_{a}}(\beta)
&=&\left\{ \begin{array}{ll}
0, &\mathrm{if}~ \beta\in\mathbb{F}_{p^{n}}\setminus\supp(\mathcal{W}_{f}), \\
\left(p^{n-1}+\epsilon\frac{p-1}{p}\sqrt{p^{*}}^{n+s}\right)\left(1-\zeta^{a}_p\right)+\epsilon\sqrt{p^{*}}^{n+s}\zeta^{a}_p, &\mathrm{if}~\beta=0, \\
\epsilon\sqrt{p^{*}}^{n+s}\frac{1}{p}\left(\zeta^{a}_p+p-1\right), & \mathrm{if}~f^*(\beta)\in\{-a,0\}, \\
\epsilon\sqrt{p^{*}}^{n+s}\frac{1}{p}\left(p\zeta_p^{f^{*}(\beta)+a}+\zeta^{a}_p-1\right), &  \mathrm{if}~f^*(\beta)\in\mathbb{F}^{\star}_p\setminus\{-a\},
\end{array} \right.
\end{eqnarray*}
if $n+s$ is even, and
\begin{eqnarray*}\label{a3}
\mathcal{W}_{f_{a}}(\beta)
&=&\left\{ \begin{array}{ll}
0, &\mathrm{if}~\beta\in\mathbb{F}_{p^{n}}\setminus\supp(\mathcal{W}_{f}), \\
p^{n-1}\left(1-\zeta^{a}_p\right)+\epsilon\sqrt{p^{*}}^{n+s}\zeta^{a}_p, & \mathrm{if}~\beta=0, \\
\epsilon\sqrt{p^{*}}^{n+s}\zeta^{a}_p, &\mathrm{if}~ f^*(\beta)=0, \\
\epsilon\sqrt{p^{*}}^{n+s}\frac{1}{p}\left(\left(1-\zeta^{a}_p\right)\eta_{0}(-a)\sqrt{p^{*}}+p\right), &\mathrm{if}~  f^*(\beta)=-a,\\
\epsilon\sqrt{p^{*}}^{n+s}\left(\frac{1}{p}\eta_{0}\left(f^{*}(\beta)\right)\sqrt{p^{*}}\left(1-\zeta^{a}_p\right)+\zeta_p^{f^*(\beta)+a}\right), & \mathrm{if}~ f^*(\beta)\in\mathbb{F}^{\star}_p\setminus\{-a\},
\end{array} \right.
\end{eqnarray*}
if $n+s$ is odd.
\end{lemma}
\begin{proof} For any $b\in \mathbb{F}_p$, we define the set $T_b=\{x\in \mathbb{F}_{ p^{n}} ~|~ f(x)=b\}.$
Then for any $\beta\in \mathbb{F}_{p^{n}}$, the Walsh transform coefficient of $f_{a}(x)$ equals
\begin{eqnarray}\label{hatf{pb0}}
\mathcal{W}_{f_{a}}(\beta)&=&\sum\limits_{x\in \mathbb{F}_{p^{n}}}\zeta_p^{f_{a}(x)-\mathrm{Tr}_1^n(\beta x)}=\sum\limits_{b \in \mathbb{F}_p}\sum\limits_{x\in T_b}\zeta_p^{b+ab^{p-1}-\mathrm{Tr}_1^n(\beta x)}\nonumber \\
&=&\sum\limits_{x\in T_0}\zeta_p^{-\mathrm{Tr}^n_1(\beta x)}+\sum\limits_{b\in\mathbb{F}^{\star}_p}\sum\limits_{x\in T_b}\zeta_p^{f(x)+a-\mathrm{Tr}^n_1(\beta x)}\nonumber \\
&=&\sum\limits_{x\in T_0}\zeta_p^{-\mathrm{Tr}^n_1(\beta x)}\left(1-\zeta^{a}_p\right)+\sum\limits_{x\in \mathbb{F}_{p^{n}}}\zeta_p^{f(x)+a-\mathrm{Tr}^n_1(\beta x)}\nonumber\\
&=&\left\{ \begin{array}{ll}
\#T_0\left(1-\zeta^{a}_p\right)+\epsilon\sqrt{p^{*}}^{n+s}\zeta_p^{a}, & \mathrm{if}~\beta=0, \\
\#0, & \mathrm{if}~\beta\in\mathbb{F}_{p^{n}}\setminus \supp(\mathcal{W}_{f}), \\
\phi+\epsilon\sqrt{p^{*}}^{n+s}\zeta_p^{f^*(\beta)+a}, & \mathrm{otherwise},
\end{array} \right.
\end{eqnarray}
where $\phi=\sum\limits_{x\in T_0}\zeta_p^{-\mathrm{Tr}^n_1(\beta x)}\left(1-\zeta^{a}_p\right)$.

We first calculate the value of $\#T_0$. It follows from the definition of unbalanced weakly regular plateaued functions that
\begin{eqnarray}\label{hatf{pb1}}
\#T_0&=&\frac{1}{p}\sum\limits_{z\in \mathbb{F}_{p}}\sum\limits_{x\in \mathbb{F}_{p^{n}}}\zeta_p^{f(x)z}
=p^{n-1}+\frac{1}{p}\sum\limits_{z\in \mathbb{F}^{\star}_p}\sigma_{z}\left(\sum\limits_{x\in \mathbb{F}_{p^{n}}}\zeta_p^{f(x)}\right)\nonumber\\
&=&\left\{ \begin{array}{ll}
p^{n-1}+\epsilon\frac{p-1}{p}\sqrt{p^{*}}^{n+s}, & \mathrm{if}~n+s~ \mathrm{is}~ \mathrm{even}, \\
p^{n-1}, & \mathrm{if}~n+s~ \mathrm{is}~ \mathrm{odd}.
\end{array} \right.
\end{eqnarray}

Next we calculate the $\mathcal{W}_{f_a}(\beta)$ for any $\beta\in \supp(\mathcal{W}_{f})$. Eq. (\ref{hatf{pb0}}) yields that
\begin{eqnarray}\label{equation-w1}
\phi&=&\sum\limits_{x\in T_{0}}\left(\zeta_p^{-\mathrm{Tr}^n_1(\beta x)}-\zeta_p^{a-\mathrm{Tr}^n_1(\beta x)}\right)\nonumber \\
&=&\frac{1}{p}\sum\limits_{z\in \mathbb{F}^{\star}_p}\sum\limits_{x\in \mathbb{F}_{p^{n}}}\zeta_p^{f(x)z-\mathrm{Tr}^n_1(\beta x)}\left(1-\zeta^{a}_p\right)+\frac{1}{p}\sum\limits_{x\in \mathbb{F}_{p^{n}}}\zeta_p^{-\mathrm{Tr}^n_1(\beta x)}\left(1-\zeta^{a}_p\right)\nonumber \\
&=&\frac{1}{p}\sum\limits_{z\in \mathbb{F}^{\star}_p}\sigma_{z}\left(\sum\limits_{x\in \mathbb{F}_{p^{n}}}\zeta_p^{f(x)-\mathrm{Tr}^n_1(z^{-1}\beta x)}\right)\left(1-\zeta^{a}_p\right)\nonumber \\
&=&\frac{1}{p}\sum\limits_{z\in \mathbb{F}^{\star}_p}(-1)^{n-1}\sigma_{z}\left(\epsilon\sqrt{p^{*}}^{n+s}\zeta_p^{f^{*}(z^{-1}\beta)}\right)\left(1-\zeta^{a}_p\right).
\end{eqnarray}

Below we shall divide the discussion into two cases according to the values of $n$. Moreover, we always assume that $l$ is a positive integer and $(l-1,p-1)=1$.

$(\mathrm{i})$ If $n+s$ is even, then it follows from Eq. (\ref{othor-prop-wp}) and $f^*(z^{-1}\beta)=z^{-l}f^*(\beta)$ that
\begin{eqnarray*}\label{hatf{pb3}}
\mathcal{W}_{f_{a}}(\beta)&=&\frac{\epsilon\sqrt{p^{*}}^{n+s}}{p}\left(\left(1-\zeta^{a}_p\right)\sum\limits_{z\in \mathbb{F}^{\star}_p}\zeta_p^{z^{1-l}f^{*}(\beta)}+p\zeta_p^{f^*(\beta)+a}\right)\nonumber\\
&=&\left\{ \begin{array}{ll}
\epsilon\sqrt{p^{*}}^{n+s}\frac{1}{p}\left(\zeta^{a}_p+p-1\right), & \mathrm{if}~f^*(\beta)\in\{-a,0\}, \\
\epsilon\sqrt{p^{*}}^{n+s}\frac{1}{p}\left(p\zeta_p^{f^{*}(\beta)+a}+\zeta^{a}_p-1\right), &  \mathrm{if}~f^*(\beta)\in\mathbb{F}^{\star}_p\setminus\{-a\}.
\end{array} \right.
\end{eqnarray*}

$(\mathrm{ii})$ If $n+s$ is odd, then it follows from Eq. (\ref{othor-prop-wp}) and Lemma \ref{equations} that
\begin{eqnarray*}\label{hatf{pb2}}
\mathcal{W}_{f_{a}}(\beta)&=&\epsilon\sqrt{p^{*}}^{n+s}\frac{1}{p}\left(\left(1-\zeta^{a}_p\right)\sum\limits_{z\in \mathbb{F}^{\star}_p}\eta_{0}(z)\zeta_p^{z^{1-l}f^{*}(\beta)}+p\zeta_p^{f^*(\beta)+a}\right)\nonumber\\
&=&\epsilon\sqrt{p^{*}}^{n+s}\frac{1}{p}\left(\left(1-\zeta^{a}_p\right)\eta_{0}(f^{*}(\beta))\sqrt{p^{*}}+p\zeta_p^{f^*(\beta)+a}\right)\nonumber\\
&=&\left\{ \begin{array}{ll}
\epsilon\sqrt{p^{*}}^{n+s}\zeta^{a}_p, &\mathrm{if}~ f^*(\beta)=0, \\
\epsilon\sqrt{p^{*}}^{n+s}\frac{1}{p}\left(\left(1-\zeta^{a}_p\right)\eta_{0}(-a)\sqrt{p^{*}}+p\right), &\mathrm{if}~  f^*(\beta)=-a,\\
\epsilon\sqrt{p^{*}}^{n+s}\left(\frac{1}{p}\eta_{0}(f^{*}(\beta))\sqrt{p^{*}}(1-\zeta^{a}_p)+\zeta_p^{f^*(\beta)+a}\right), & \mathrm{if}~ f^*(\beta)\in\mathbb{F}^{\star}_p\setminus\{-a\}.
\end{array} \right.
\end{eqnarray*}
Hence the desired results follow by the above arguments. \hfill $\square$
\end{proof}

The next result plays an initial role in calculating the parameters of $\mathcal{C}^{\star}_{f_a}$ when $n+s$ is an odd positive integer.
\begin{lemma}\label{sqnsql} Let $c$ run through all elements of $\mathbb{F}^{\star}_{p}\setminus\{-a\}$. Then we have
\begin{eqnarray*}\label{sqin}
\eta_0(c)+\eta_0(c+a)
=\left\{ \begin{array}{ll}
2, &\frac{p-3-\eta_0(a)(1+\eta_0(-1))}{4} ~\mathrm{times},\\
0_{\textnormal{SQ}}, &\frac{p-1+\eta_0(a)\left(1-\eta_0(-1)\right)}{4} ~\mathrm{times},\\
0_{\textnormal{NSQ}}, &\frac{p-1-\eta_0(a)(1-\eta_0(-1))}{4} ~\mathrm{times},\\
-2, &\frac{p-3+\eta_0(a)(1+\eta_0(-1))}{4} ~\mathrm{times},
\end{array} \right.
\end{eqnarray*}
where $\eta_0(c)+\eta_0(c+a)=0_{\textnormal{SQ}}$ means that $c$ belongs to $\textnormal{SQ}$ and satisfies the equation  $\eta_0(c)+\eta_0(c+a)=0$, and $\eta_0(c)+\eta_0(c+a)=0_{\textnormal{NSQ}}$ means that $c$ belongs to $\textnormal{NSQ}$ and satisfies the equation $\eta_0(c)+\eta_0(c+a)=0$.
\end{lemma}
\begin{proof} Let $A$ be the number of $c\in\mathbb{F}^{\star}_{p}\setminus\{-a\}$ such that $\eta_0(a)+\eta_0(c+a)=2$. Then we have
\begin{eqnarray*}\label{c-number}
A&=&\sum\limits_{c\in\mathbb{F}^{\star}_{p}\setminus\{-a\}}\frac{1+\eta_0(c)}{2}\frac{1+\eta_0(c+a)}{2}\nonumber\\
&=&\frac{1}{4}(p-2+\sum\limits_{c\in\mathbb{F}^{\star}_{p}\setminus\{-a\}}\eta_0(c)+
\sum\limits_{c\in\mathbb{F}^{\star}_{p}\setminus\{-a\}}\eta_0(c+a)+\sum\limits_{c\in\mathbb{F}^{\star}_{p}\setminus\{-a\}}\eta_0(c^2+ac))\nonumber\\
&=&\frac{1}{4}\left(p-2-\eta_0(-a)-\eta_0(a)+\sum\limits_{c\in\mathbb{F}_{p}}\eta_0(c^2+ac)\right)\nonumber\\
&=&\frac{p-3-\eta_0(a)(1+\eta_0(-1))}{4},
\end{eqnarray*}
where the third equality holds by $\sum_{c\in\mathbb{F}^{\star}_{p}}\eta_0(c)=0$ and the fourth equality holds by Lemma \ref{squar-Fact}. The rest of the argument can be obtained similarly. \hfill$\square$
\end{proof}

The following result from \cite[Lemma 10]{Mesnager-2020} will contribute to the determination of the weight distribution of $\mathcal{C}^{\star}_{f_a}$ in the sequal.
\begin{lemma}\label{lemma-s}
Let notation be the same as before and $f(x)$ be a weakly regular plateaued function defined on $\mathbb{F}_{p^n}$. For any $z\in\mathbb{F}_{p}$, define
$$N_{f^{*}}(z)=\#\left\{\beta\in \supp(\mathcal{W}_{f})\ | \ f^*(\beta)=z\right\}.$$
Then the following statements hold.

$(1)$ If $n-s$ is even, then we get
\begin{eqnarray*}
N_{f^{*}}(z)&=&\left\{ \begin{array}{ll}
p^{n-s-1}+\epsilon\eta^{n+\frac{n-s}{2}}_0(-1) p^{\frac{n-s-2}{2}}(p-1), &\mathrm{if}~z=0,\\
p^{n-s-1}-\epsilon\eta^{n+\frac{n-s}{2}}_0(-1) p^{\frac{n-s-2}{2}}, &\mathrm{if}~z\in \mathbb{F}^{\star}_{p}.
\end{array} \right.
\end{eqnarray*}

$(2)$ If $n-s$ is odd, then we get
\begin{eqnarray*}
N_{f^{*}}(z)&=&\left\{ \begin{array}{ll}
p^{n-s-1}, &\mathrm{if}~z=0,\\
p^{n-s-1}+\epsilon\eta^{n+\frac{n-s-1}{2}}_0(-1) p^{\frac{n-s-1}{2}}, &\mathrm{if}~z\in \textnormal{SQ},\\
p^{n-s-1}-\epsilon\eta^{n+\frac{n-s-1}{2}}_0(-1) p^{\frac{n-s-1}{2}}, &\mathrm{if}~z\in \textnormal{NSQ}.
\end{array} \right.
\end{eqnarray*}
\end{lemma}
\begin{theorem}\label{ccww}
Let $n,s$ be positive integers and $n-s>2$. Let $\mathcal{C}^{\star}_{f_{a}}$ be a linear code derived from the $f_{a}(x)$ in \textnormal{Eq. (\ref{pfunction})}. Then the following statements hold.

$(1)$ If $n+s$ is even, then $\mathcal{C}^{\star}_{f_{a}}$ is a $\left[p^{n}-1,n+1, (p-2)\left(p^{n-1}-\eta^{\frac{n+s}{2}}_0(-1)p^{\frac{n+s}{2}-1}\right)\right]$ self-orthogonal minimal four-weight code with $\frac{w_{\min}}{w_{\max}}<\frac{p-1}{p}$ and the weight distribution of $\mathcal{C}^{\star}_{f_{a}}$ is given by the following \textnormal{Table \ref{leven-1}}.

$(2)$ If $n+s$ is odd, then $\mathcal{C}^{\star}_{f_{a}}$ is a $\left[p^{n},n+1, (p-2)p^{n-1}-\eta_0(a)\epsilon\eta^{\frac{n+s+1}{2}}_0(-1)p^{\frac{n+s-1}{2}}\right]$ self-orthogonal minimal code with $\frac{w_{\min}}{w_{\max}}<\frac{p-1}{p}$ and the weight distribution of $\mathcal{C}^{\star}_{f_{a}}$ is given by the following \textnormal{Table \ref{odd-1}}.
\end{theorem}
\begin{table}[H]
\begin{center}
\caption{The weight distribution of $\mathcal{C}^{\star}_{f_{a}}$ when $n+s$ is even}\label{leven-1}
\begin{tabular}{ll}
\hline\noalign{\smallskip}
Weight  &  Multiplicity   \\
\noalign{\smallskip}
\hline\noalign{\smallskip}
$0$  &  1 \\
$ (p-1)p^{n-1}$    &  $(p-1)\left(p^n-p^{n-s}\right)+p^n-1$   \\
$ (p-2)\left(p^{n-1}-\epsilon\eta^{\frac{n+s}{2}}_0(-1)p^{\frac{n+s}{2}-1}\right)$  &  $ p-1$     \\
$ (p-1)p^{n-1}-\epsilon\eta^{\frac{n+s}{2}}_0(-1)(p-2)p^{\frac{n+s}{2}-1}$  &  $(p-1)\left(2p^{n-s-1}+\epsilon\eta^{n+\frac{n-s}{2}}_0(-1)(p-2)p^{\frac{n-s}{2}-1}-1\right) $    \\
$ (p-1)p^{n-1}+\epsilon\eta^{\frac{n+s}{2}}_0(-1)2p^{\frac{n+s}{2}-1}$  &  $ (p-1)(p-2)\left(p^{n-s-1}-\epsilon\eta^{n+\frac{n-s}{2}}_0(-1) p^{\frac{n-s}{2}-1}\right)$     \\
\noalign{\smallskip}
\hline
\end{tabular}
\end{center}
\end{table}
\vspace{-0.5cm}
\begin{table}[H]
\begin{center}
\caption{The weight distribution of $\mathcal{C}^{\star}_{f_{a}}$ when $n+s$ is odd}\label{odd-1}
\begin{tabular}{ll}
\hline\noalign{\smallskip}
Weight  &  Multiplicity   \\
\noalign{\smallskip}
\hline\noalign{\smallskip}
$0$  &  1 \\
$ (p-1)p^{n-1}$  &  $ \frac{(p-1)^{2}p^{n-s-1}+(p-1)\eta_0(a)(1-\eta_0(-1))\varepsilon_1' } {2}+p^{n+1}-(p-1)p^{n-s}-1$     \\
$ (p-2)p^{n-1}-\varepsilon'\eta_0(a)$    &  $   p-1$   \\
$ (p-1)p^{n-1}-\varepsilon'\eta_0(a)$  &  $(p-1)\left(p^{n-s-1}-1\right)$    \\
$(p-1)p^{n-1}-\varepsilon'\eta_0(-a)$  &  $ (p-1)\left(p^{n-s-1}+\varepsilon_1'\eta_0(-a)\right)$     \\
$ (p-1)p^{n-1}+2\varepsilon'\eta_0(a)$  &  $(p-1)\frac{p-3+\eta_0(a)(1+\eta_0(-1))}{4}\left(p^{n-s-1}-\varepsilon_1' \right) $     \\
$ (p-1)p^{n-1}-2\varepsilon'\eta_0(a)$  &  $(p-1)\frac{p-3-\eta_0(a)(1+\eta_0(-1))}{4}\left(p^{n-s-1}+\varepsilon_1' \right) $     \\
\noalign{\smallskip}
\hline
\end{tabular}
\end{center}
where $\varepsilon'=\epsilon\eta^{\frac{n+s+1}{2}}_0(-1)p^{\frac{n+s-1}{2}},\varepsilon_1'=\epsilon\eta^{n+\frac{n-s-1}{2}}_0(-1)p^{\frac{n-s-1}{2}}.$
\end{table}
\vspace{-0.5cm}
\begin{proof} We first prove the statement (1). By applying Lemma \ref{lemma-8} to Lemma \ref{code-lemma}, it is straightforward to show that $\wt(\mathbf{c}_{\alpha,\beta})$ belongs to
$$\Big\{(p-1)p^{n-1}, (p-2)p^{n-1}, (p-1)p^{n-1}-\epsilon\eta^{\frac{n+s}{2}}_0(-1)(p-2)p^{\frac{n+s}{2}-1},$$
$$ (p-1)p^{n-1}+\epsilon\eta^{\frac{n+s}{2}}_0(-1)2p^{\frac{n+s}{2}-1}\Big\}.$$
Obviously, the $\mathcal{C}_{f_{a}}$ is a four-weight code. It is evident to see that
\begin{eqnarray*}
\frac{w_{\min}}{w_{\max}}&=& \frac{(p-2)\left(p^{n-1}-\epsilon\eta^{\frac{n+s}{2}}_0(-1)p^{\frac{n+s}{2}-1}\right)}{(p-1)p^{n-1}+\left(\frac{1-\epsilon\eta^{\frac{n+s}{2}}_0(-1)}{2}p+\epsilon\eta^{\frac{n+s}{2}}_0(-1)2\right)p^{\frac{n+s}{2}-1}}\\
&=&1-\frac{p^{n-1}+\frac{p\left(1+\epsilon\eta^{\frac{n+s}{2}}_0(-1)\right)p^{\frac{n+s}{2}-1}}{2}}{(p-1)p^{n-1}+\left(\frac{1-\epsilon\eta^{\frac{n+s}{2}}_0(-1)}{2}p+\epsilon\eta^{\frac{n+s}{2}}_0(-1)2\right)p^{\frac{n+s}{2}-1}}<\frac{p-1}{p}.
\end{eqnarray*}

Therefore, Lemma \ref{LEM1} cannot be used to determine whether the linear code $\mathcal{C}^{\star}_{f_{a}}$ is minimal. To prove minimality, we denote by $\omega_1$ to $\omega_4$ the elements in the set above from left to right. Applying Lemma \ref{lemma-8} to Lemma \ref{code-lemma} again, we conclude that the sets of codewords are as follows.
\begin{align}
D_1&=\left\{\mathbf{c}_{0,\beta}\in\mathcal{C}_{f_{a}}~|~\alpha=0,\beta\ne0,~\mathrm{or}~\alpha\ne0,\beta\in\mathbb{F}_{p^{n}}\setminus\supp(\mathcal{W}_{f})\right\},\nonumber \\
D_2&=\left\{\mathbf{c}_{\alpha,0}\in\mathcal{C}_{f_{a}}~|~\alpha\ne0,\beta=0\right\},\nonumber\\
D_3&=\left\{\mathbf{c}_{\alpha,\beta}\in\mathcal{C}_{f_{a}}~|~\alpha \beta\ne0,~\beta\in\supp(\mathcal{W}_{f}),~  f^*(\alpha^{-1}\beta)\in\{-a,0\}\right\},\nonumber \\
D_4&=\left\{\mathbf{c}_{\alpha,\beta}\in\mathcal{C}_{f_{a}}~|~\alpha\beta\ne0,~\beta\in\supp(\mathcal{W}_{f}),~f^*(\alpha^{-1}\beta)\in\mathbb{F}^*_p\setminus\{-a\}\right\}.\nonumber \end{align}
Combining Lemma \ref{tffunction g} and Lemma \ref{lemma-s}, in light of the definitions of $D_1$ to $D_{4}$, we can arrive at Table \ref{leven-1} after a tedious calculation.
Actually, it suffices to consider the case of $\epsilon\eta^{\frac{n+s}{2}}_0(-1)=1$, since the case of $\epsilon\eta^{\frac{n+s}{2}}_0(-1)=-1$ can be dealt with similarly. Note that the weights of codewords in sets $D_1$ to $D_4$ correspond to $\omega_1$ to $\omega_4$, respectively. Thus Table \ref{leven-1} yields that
\begin{equation}\label{weight-leq}
\omega_4>\omega_1>\omega_3>\omega_2.
\end{equation}
Let $(\alpha_j,\beta_j)\in\mathbb{F}_p\times\mathbb{F}_{p^{n}}\setminus\{(0,0)\}$ with $1\leq j\leq2$. Then for any two linearly independent codewords $\mathbf{a}=\mathbf{c}_{\alpha_1,\beta_1}$ and $\mathbf{b}=\mathbf{c}_{\alpha_2,\beta_2}$, the coverage of codewords in $\mathcal{C}_{f_{a}}$ can be divided into the following four cases.

$(\mathrm{i})$ $\mathbf{a}=\mathbf{c}_{\alpha_1, \beta_1}\in D_1$.

$(\mathrm{i}.1)$ Let $\mathbf{b}=\mathbf{c}_{\alpha_2, 0}\in D_2$. It is easy to see that $\mathbf{a}+z\mathbf{b}=\mathbf{c}_{\alpha_1+z\alpha_2,\beta_1}\in  D_1\cup D_3\cup D_4$ when $z$ runs through $\mathbb{F}^{\star}_p$. Thus Eq. (\ref{weight-leq}) yields that
$$\sum_{z\in\mathbb{F}^{\star}_p}\wt(\mathbf{a}+z\mathbf{b})\geq(p-1)\omega_3=(p-1)\left(p^n-p^{n-1}-(p-2)p^{\frac{n+s}{2}-1}\right),$$
$$(p-1)\wt(\mathbf{a})-\wt(\mathbf{b})=(p-1)\omega_1-\omega_2=(p-1)^2p^{n-1}-(p-2)\left(p^{n-1}-p^{\frac{n+s}{2}-1}\right).$$
Obviously, we can arrive at
$$\sum\limits_{z \in \mathbb{F}^{\star}_p}\wt(\mathbf{a}+z\mathbf{b})>(p-1)\wt(\mathbf{a})-\wt(\mathbf{b}).$$

$(\mathrm{i}.2)$ Let $\mathbf{b}=\mathbf{c}_{\alpha_2,\beta_2}\in D_1\cup D_3\cup D_4$.  It is clear that $\mathbf{a}+z\mathbf{b}=\mathbf{c}_{\alpha_1+z\alpha_2,\beta_1+z\beta_2}\in \bigcup_{i=1}^4 D_i$ when $z$ runs through $\mathbb{F}^{\star}_p$. By the definition of $D_2$, it can be deduced that at most only one codeword belongs to $D_2$ in the set $\{\mathbf{a}+z\mathbf{b}=\mathbf{c}_{z\alpha_2,\beta_1+z\beta_2}~|~z\in\mathbb{F}^{\star}_p\}$ for given codewords $\mathbf{a}$ and $\mathbf{b}$. Then it follows from Eq. (\ref{weight-leq}) that
$$\sum_{z\in\mathbb{F}^{\star}_p}\wt(\mathbf{a}+z\mathbf{b})\geq(p-2)\omega_3+\omega_2=(p-2)p^n-(p-2)(p-1)p^{\frac{n+s}{2}-1},$$
$$(p-1)\wt(\mathbf{a})-\wt(\mathbf{b})\leq(p-1)\omega_1-\omega_3=(p-1)(p-2)p^{n-1}-(p-2)p^{\frac{n+s}{2}-1}.$$
 Clearly, we get
$$\sum\limits_{z \in \mathbb{F}^{\star}_p}\wt(\mathbf{a}+z\mathbf{b})>(p-1)\wt(\mathbf{a})-\wt(\mathbf{b}).$$

$(\mathrm{ii})$ $\mathbf{a}=\mathbf{c}_{\alpha_1,0}\in D_2$. For any $\mathbf{b}\in \bigcup_{i=1}^4 D_i$, from Eq. (\ref{weight-leq}) we get
$$
\sum_{z\in\mathbb{F}^{\star}_p}\wt(\mathbf{a}+z\mathbf{b})\geq(p-1)\omega_2>(p-1)\omega_2-\omega_2\geq(p-1)\wt(\mathbf{a})-\wt(\mathbf{b}).$$

$(\mathrm{iii})$ $\mathbf{a}=\mathbf{c}_{\alpha_1,\beta_1}\in D_3$.

$(\mathrm{iii}.1)$ Let $\mathbf{b}=\mathbf{c}_{\alpha_2, 0}\in D_2$. It is evident to see that $\mathbf{a}+z\mathbf{b}=\mathbf{c}_{\alpha_1+z\alpha_2,\beta_1}\in  D_1\cup D_3\cup D_4$ when $z$ runs through $\mathbb{F}^{\star}_p$. By the definition of $D_1$, it can be deduced that at least one codeword belongs to $D_1$ in the set $\{\mathbf{a}+z\mathbf{b}=\mathbf{c}_{\alpha_1+z\alpha_2,\beta_1}~|~z\in\mathbb{F}^{\star}_p\}$
for given codewords $\mathbf{a}$ and $\mathbf{b}$. Consequently, by Eq. (\ref{weight-leq}) we have
$$\sum_{z\in\mathbb{F}^{\star}_p}\wt(\mathbf{a}+z\mathbf{b})\geq(p-2)\omega_3+\omega_1=(p-1)^2p^{n-1}-(p-2)^2p^{\frac{n+s}{2}-1},$$
$$(p-1)\wt(\mathbf{a})-\wt(\mathbf{b})=(p-1)\omega_3-\omega_2=\left(p^2-3p+3\right)p^{n-1}-(p-2)p^{\frac{n+s}{2}}.$$
Evidently,
$$\sum\limits_{z \in \mathbb{F}^{\star}_p}\wt(\mathbf{a}+z\mathbf{b})>(p-1)\wt(\mathbf{a})-\wt(\mathbf{b}).$$

$(\mathrm{iii}.2)$ Let $\mathbf{b}=\mathbf{c}_{\alpha_2,\beta_2}\in D_1\cup D_3\cup D_4$. A proof similar to that of $(\mathrm{i}.2)$ yields that
$$\sum_{z\in\mathbb{F}^{\star}_p}\wt(\mathbf{a}+z\mathbf{b})\geq(p-2)\omega_3+\omega_2=(p-2)p^n-(p-2)(p-1)p^{\frac{n+s}{2}-1},$$
$$(p-1)\wt(\mathbf{a})-\wt(\mathbf{b})\leq(p-1)\omega_3-\omega_3=(p-2)(p-1)p^{n-1}-(p-2)^2p^{\frac{n+s}{2}-1}.$$
 Clearly, we get
$$\sum\limits_{z \in \mathbb{F}^{\star}_p}\wt(\mathbf{a}+z\mathbf{b})>(p-1)\wt(\mathbf{a})-\wt(\mathbf{b}).$$

$(\mathrm{iv})$ $\mathbf{a}=\mathbf{c}_{\alpha_1,\beta_1}\in D_4$.
For any $\mathbf{b}\in\bigcup_{i=1}^4 D_i$, by the same method as in the proof of $(\mathrm{iii})$, we can easily get that
$$\sum\limits_{z \in \mathbb{F}^{\star}_p}\wt(\mathbf{a}+z\mathbf{b})>(p-1)\wt(\mathbf{a})-\wt(\mathbf{b}).$$

Finally, based on the above all discussions and Lemma \ref{nsc}, we get that $\mathcal{C}^{\star}_{f_{a}}$ is a minimal linear code.

We now proceed to show the statement (2). Applying Lemma \ref{lemma-8} to Lemma \ref{code-lemma}, and then in light of Lemma \ref{equations} and $\sigma_a(\sqrt {p^*})=\eta_0(a)\sqrt{p^*}$ with $a\in\mathbb{F}^{\star}_p$, we can deduce that
\begin{align}
D'_1&=\{\mathbf{c}_{\alpha,0}\in\mathcal{C}_{f_{a}}~|~\alpha=0,\beta\ne0;~\alpha\ne0,\beta\in\mathbb{F}_{p^{n}}\setminus\supp(\mathcal{W}_{f});\nonumber\\&~~~~~~\alpha \beta\ne0,~\beta\in\supp(\mathcal{W}_{f}),~f^*(\alpha^{-1}\beta)(f^*(\alpha^{-1}\beta)+a)\in \textnormal{NSQ}\},\nonumber \\
D'_2&=\left\{\mathbf{c}_{\alpha,0}\in\mathcal{C}_{f_{a}}~|~\alpha\ne0,\beta=0\right\},\nonumber \\
D'_3&=\left\{\mathbf{c}_{\alpha,\beta}\in\mathcal{C}_{f_{a}}~|~\alpha \beta\ne0,~\beta\in\supp(\mathcal{W}_{f}),~ f^*(\alpha^{-1}\beta)=0\right\},\nonumber \\
D'_4&=\left\{\mathbf{c}_{\alpha,\beta}\in\mathcal{C}_{f_{a}}~|~\beta\in\supp(\mathcal{W}_{f}),~f^*(\alpha^{-1}\beta)=-a\right\},\nonumber \\
D'_5&=\left\{\mathbf{c}_{\alpha,\beta}\in\mathcal{C}_{f_{a}}~|~\beta\in\supp(\mathcal{W}_{f}),~f^*(\alpha^{-1}\beta)\in \textnormal{NSQ},f^*(\alpha^{-1}\beta)+a\in \textnormal{NSQ}\right\},\nonumber \\
D'_6&=\left\{\mathbf{c}_{\alpha,\beta}\in\mathcal{C}_{f_{a}}~|~\beta\in\supp(\mathcal{W}_{f}),~f^*(\alpha^{-1}\beta)\in \textnormal{SQ},f^*(\alpha^{-1}\beta)+a\in \textnormal{SQ}\right\}.\nonumber
\end{align}
Based on Lemmas \ref{tffunction g}, \ref{sqnsql}, \ref{lemma-s}, and the definitions of $D'_1$ to $D'_{6}$, we can conclude the results shown in Table \ref{odd-1}. Then similar to the method used to prove the minimality of the linear code in the statement (1), by dividing into six cases based on the sets $D'_j$ that $\mathbf{c}_{\alpha,\beta}$ belongs to, $1\leq j\leq6$, we obtain that any two linearly independent codewords in $\mathcal{C}^{\star}_{f_{a}}$ satisfy Lemma \ref{nsc}.
Finally, from Eqs. (\ref{code_1}), (\ref{vzcode_1}) and Lemma \ref{p-ary}, we see that if $f(0)=0$, then the self-orthogonality of the codes $\mathcal{C}_{f}$ and $\mathcal{C}^{\star}_{f}$ is equivalent. Thus, it follows from Theorem \ref{bent function1}, Lemmas \ref{3-ary} and \ref{lemma-8} that $\mathcal{C}^{\star}_{f_{a}}$ is self-orthogonal. \hfill $\square$
\end{proof}
\begin{example}\label{examples linear code}
Let $p=3$, $n=5$, $a=1$ and $\xi$ be a primitive element in $\mathbb{F}_{3^5}$ such that $\xi^5+2\xi+1=0$.

(1) Let $f(x)=\mathrm{Tr}^{n}_{1}\left(x^{2}+\xi^23x^{4}+\xi^4x^{10}\right)$. Then $f(x)$ is a quadratic unbalanced weakly regular ternary 1-plateaued function with
 $$\mathcal{W}_{f}(\beta)\in\left\{0, 27, 27\zeta_3, 27\zeta^2_3\right\}$$
for all $\beta\in\mathbb{F}_{3^5}$, where $\epsilon=1$. Furthermore, we can verify by Magma
program that the linear code $\mathcal{C}^{\star}_{f_1}$ has the parameters $[242, 6, 72]$ and the weight enumerator is $1+2z^{72}+112z^{153}+566z^{162}+48z^{180}$ with $w_{\min}/w_{\max}=72/180<2/3$. This is consistent with Theorem \ref{ccww}(1).

(2) Let $f(x)=\mathrm{Tr}^{n}_{1}\left(\xi x^{2}+x^{4}+2x^{10}\right)$. Then $f(x)$ is a quadratic unbalanced weakly regular ternary 2-plateaued function with
 $$\mathcal{W}_{f}(\beta)\in\left\{0, -27\sqrt{3}i, -27\sqrt{3}\zeta_3, -27\sqrt{3}\zeta^2_3\right\}$$
for all $\beta\in\mathbb{F}_{3^5}$, where $\epsilon=1$ and $i=\sqrt{-1}$. Furthermore, we can verify by Magma
program that the linear code $\mathcal{C}^{\star}_{f_1}$ has the parameters $[242, 6, 54]$ and the weight enumerator is $1+2z^{54}+16z^{135}+698z^{162}+12z^{189}$ with $w_{\min}/w_{\max}=54/189<2/3$. This is consistent with Theorem \ref{ccww}(2).
\end{example}
\begin{remark}\label{lttremark} (1) By observing (i.1) in the proof of Theorem \ref{ccww}, we deduce that $(p-1)\omega_3=(p-1)\omega_1-\omega_2$ for $n-s=2$. Thus, in light of Lemma \ref{nsc}, it follows that the code $\mathcal{C}^{\star}_{f_a}$ is not minimal when $n-s=2$. Therefore, we assume that $n-s>2$ in Theorem \ref{ccww}.

(2) Combining Theorem \ref{bent function1} with Lemma \ref{lemma-8}, we can derive that $\bar{\mathcal{C}}_{f_{a}}$ is self-orthogonal. Although applying Lemma \ref{lemma-8} into Lemma \ref{1code-lemma} yields the weight distribution of $\bar{\mathcal{C}}_{f_{a}}$, we could not do this, as it is not a minimal code. Moreover, by observing Table \ref{odd-1}, we see that if $p=3$ or $p=5$, then $\mathcal{C}^{\star}_{f_{a}}$ is a self-orthogonal minimal four-weight code with $\frac{w_{\min}}{w_{\max}}<\frac{p-1}{p}$; if $p \equiv1\pmod4$ with $p>5$, then $\mathcal{C}^{\star}_{f_{a}}$ is a self-orthogonal minimal five-weight code with $\frac{w_{\min}}{w_{\max}}<\frac{p-1}{p}$; otherwise, it is a self-orthogonal minimal six-weight code with $\frac{w_{\min}}{w_{\max}}<\frac{p-1}{p}$.
\end{remark}
\section{Conclusions}\label{Sec-Conclusion}

In this paper, we have mainly investigated the constructions of self-orthogonal minimal linear codes violating the AB condition. We have provided the necessary and sufficient conditions for the linear codes designed from Construction 1 to be self-orthogonal. More precisely, in Theorem \ref{lfunction1}, the necessary and sufficient conditions for determining whether a binary linear code is self-orthogonal have been provided. As an application, two classes of minimal binary codes that violate the AB condition have been constructed, and Theorem \ref{lfunction1} determined their self-orthogonality. In Theorem \ref{bent function1}, we also presented the necessary and sufficient conditions for determining the self-orthogonality of a $p$-ary linear code. Subsequently, one class of two-weight and one class of four-weight linear codes, in addition to two classes of minimal codes violating the AB condition, have been designed. The conditions under which these codes are self-orthogonal have been clearly established through Theorem \ref{bent function1}. The constructed codes have different parameters from the existing literature and are inequivalent to the known ones as far as we know.
Furthermore, the minimal linear codes in this paper can be used for secret sharing schemes, and secure two-party computation with the framework introduced in \cite{CDY,CCP,YD}. As future work, it would be interesting to exhibit more classes of self-orthogonal minimal codes that violate the AB condition.

%%%%%%%%%%%%%%%%%%%%%%%%%%%%

\section*{Acknowledgments}
The work was supported by the National Key Research and Development Program of China under Grant (No. 2024YFA1013000 and No. 2022YFA1004900). The work of Kangquan Li was supported by the Research Fund of National University of Defense Technology under Grant ZK22-14 and the National Natural Science Foundation of China (NSFC) under Grant (No. 62202476 and No. 62172427). The work of Longjiang Qu was supported by the National Natural Science Foundation of China under Grant (No. 62032009).

\end{document}